\documentclass[11pt,a4paper]{article}

\usepackage{amsfonts}
\usepackage{amsmath}
\usepackage{amssymb}
\usepackage{amsthm}
\usepackage{a4wide}
\usepackage{bbm}
\usepackage{fullpage}
\usepackage{mleftright}
\usepackage{url}
\usepackage{verbatim}

\newcommand*\N{\mathbb{N}}
\newcommand*\R{\mathbb{R}}
\newcommand*\cupdot{\mathbin{\mathaccent\cdot\cup}}
\newcommand*\inproduct[2]{\left\langle#1,#2\right\rangle}
\newcommand*\norm[1]{\left\lVert #1 \right\rVert}
\newcommand*\one{\mathbf 1}
\newcommand*\rfrac[2]{{}^{#1\!\!}/_{\!#2}}
\newcommand*\vol{\mbox{vol}}
\mleftright

\numberwithin{equation}{section}
\newtheorem{thm}{Theorem}[section]
\newtheorem{pro}[thm]{Proposition}
\newtheorem{lem}[thm]{Lemma}

\theoremstyle{definition}
\newtheorem{dfn}[thm]{Definition}
\newtheorem{rem}[thm]{Remark}
\newtheorem{exm}[thm]{Example}

\begin{document}

\title{High Dimensional Random Walks and Colorful Expansion}
\author{
	Tali Kaufman\thanks{Bar-Ilan University, ISRAEL. Email: \texttt{kaufmant@mit.edu}. Research supported in part by ERC and BSF.} \and
	David Mass\thanks{Bar-Ilan University, ISRAEL. Email: \texttt{dudimass@gmail.com}.}}
\maketitle

\begin{abstract}
	Random walks on bounded degree expander graphs have numerous applications, both in theoretical and practical computational problems. A key property of these walks is that they converge rapidly to their stationary distribution.
	
	In this work we {\em define high order random walks}: These are generalizations of random walks on graphs to high dimensional simplicial complexes, which are the high dimensional analogues of graphs. A simplicial complex of dimension $d$ has vertices, edges, triangles, pyramids, up to $d$-dimensional cells. For any $0 \leq i < d$, a high order random walk on dimension $i$ moves between neighboring $i$-faces (e.g., edges) of the complex, where two $i$-faces are considered neighbors if they share a common $(i+1)$-face (e.g., a triangle). The case of $i=0$ recovers the well studied random walk on graphs.
	
	We provide a {\em local-to-global criterion} on a complex which implies {\em rapid convergence of all high order random walks} on it. Specifically, we prove that if the $1$-dimensional skeletons of all the links of a complex are spectral expanders, then for {\em all} $0 \le i < d$ the high order random walk on dimension $i$ converges rapidly to its stationary distribution.
	
	We derive our result through a new notion of high dimensional combinatorial expansion of complexes which we term {\em colorful expansion}. This notion is a natural generalization of combinatorial expansion of graphs and is strongly related to the convergence rate of the high order random walks.
	
	We further show an explicit family of {\em bounded degree} complexes which satisfy this criterion. Specifically, we show that Ramanujan complexes meet this criterion, and thus form an explicit family of bounded degree high dimensional simplicial complexes in which all of the high order random walks converge rapidly to their stationary distribution.
\end{abstract}

\section{Introduction}
Expander graphs play a fundamental role in computer science, mathematics, physics and more (see e.g.~\cite{HLW06} and~\cite{Lub12}). In particular, random walks on bounded degree expander graphs have been proven very useful for many applications. A key property of these walks is that they converge rapidly to their stationary distribution. This property allows one to efficiently sample points from some huge space while investing only a few random bits, which is necessary in many theoretical and practical computational problems.

In recent years a generalization of expander graphs to higher dimensions has emerged. The high dimensional analogue of a graph is called a {\em simplicial complex}. It can be viewed as a hypergraph with a closure property, namely, for any hyperedge in the hypergraph all of its subsets are also in the hypergraph. An hyperedge of size $i$ is called an {\em $(i-1)$-face} of the complex. Thus, the vertices of the complex are the $0$-faces, the edges are the $1$-faces, the triangles are the $2$-faces, etc. A complex is termed {\em bounded degree} if the number of faces of any dimension which are incident to any vertex is bounded by a constant, independent of the number of vertices in the complex.

\paragraph{On high dimensional expanders.} The study of high dimensional expanders gains a lot of interest recently. They were proven to be a bridge between deep questions in topology (such as topological overlapping, see~\cite{Gro10}) and important questions in theoretical computer science (such as property testing and error-correcting codes, see~\cite{KL14}). The interested reader is referred to~\cite{Lub14} for a recent survey on high dimensional expanders. The general hope is that high dimensional expanders could be used in places where expander graphs were not good enough.

Expander graphs have had an enormous effect on computer science. However, there are some striking applications in which expanders yield good results but not good enough. For example, by the celebrated result of Sipser and Spielman~\cite{SS96}, expanders imply good error-correcting codes. However, the codes obtained by~\cite{SS96} are not locally testable, namely, the very desired property of local testability (which in fact was the motivating goal in the construction of error-correcting codes from expanders, as discussed in Spielman's thesis~\cite{Spi95}) is not known to be achieved by codes based on expanders. 

Another canonical example where we have seen the limit of expanders is the new proof of the PCP theorem obtained by Dinur~\cite{Din07}. Dinur's proof builds on properties of expanders to yield a simpler proof of the celebrated PCP theorem. However, the PCP obtained by expanders is not strong enough to yield good inapproximation results.

The thesis we try to pursue here is that high dimensional expanders could potentially be used in applications as above, where (one-dimensional) expanders were not strong enough. We expect high dimensional expanders to yield better results than the well studied (one-dimensional) expanders since they are stronger objects: They expand at {\em all} dimensions, i.e., with respect to vertices, edges, triangles, etc., and not just with respect to vertices as expander graphs.

Dinur's proof of the PCP theorem builds heavily on random walks on expander graphs. In the following we study high dimensional analogues of random walks in high dimensional simplicial complexes.

\paragraph{On high dimensional random walks.} In this paper we define high order random walks on simplicial complexes. Our definition generalizes the notion of random walks on graphs to high dimensional simplicial complexes.

For a $d$-dimensional simplicial complex, we define the high order random walk on it for any dimension $0 \le i < d$. This walk moves at random between neighboring $i$-faces (e.g., edges) of the complex, where two $i$-faces are considered neighbors if they share a common $(i+1)$-face (e.g., a triangle).

The question we address in this paper is the following.

\paragraph{Question.}{\em Are there bounded degree high dimensional simplicial complexes in which all of the high order random walks converge rapidly to their stationary distribution}.
\bigskip

In a nutshell, our answer is the following. We provide a local criterion on a complex which implies the rapid convergence of all high order random walks on it. Moreover, we show an explicit family of bounded degree complexes which satisfy that criterion, and hence are bounded degree high dimensional simplicial complexes in which all of the high order random walks converge rapidly to their stationary distribution.

The convergence rate of random walks on graphs is known to be deduced from the spectrum of the graph's adjacency matrix (for more about random walks on graphs see~\cite{Lov93}). Similarly, for a $d$-dimensional simplicial complex and any $0 \le i < d$ we define the matrix $A_i$ to be the adjacency matrix of dimension $i$ of the complex. This matrix is indexed by the $i$-faces of the complex, where $A_i(\sigma,\tau)$ indicate if $\sigma$ and $\tau$ are neighbors. The question we address here is how one could construct a bounded degree $d$-dimensional simplicial complex (for any $d \in \N$) whose {\em all} adjacency matrices $A_i$, $0 \le i < d$, have a concentrated spectrum, and thus its associated high order random walks (on any of its dimensions) converge rapidly to their stationary distribution.

\paragraph{Bounded versus unbounded.} It is easy to obtain an unbounded degree graph in which the random walk will converge rapidly to its stationary distribution. This case is less useful as in most applications it is crucial to have a constant bound on the number of neighbors of any vertex in the graph. In a same way, the case of having an unbounded degree complex so the high order random walks on it will converge rapidly to their stationary distribution is less interesting. The complete high dimensional analogue of the most useful random walks on expander graphs are {\em bounded degree} high dimensional simplicial complexes in which all of the high order random walks converge rapidly to their stationary distribution.

\paragraph{Better than random.} In the case of graphs, a random walk on a random bounded degree graph is known to converge rapidly to its stationary distribution. The constructions of explicit expander graphs seem to try and simulate the behavior of these random objects as much as possible, i.e., the philosophy is that a random graph is the best one can get and with explicit constructions we try to get as close as we can to random graphs. In the case of high dimensional simplicial complexes, a high order random walk on a random bounded degree complex {\em does not} converge rapidly to its stationary distribution. Thus, this phenomenon we study here is very special as it is something we {\em cannot achieve with a random object}.

\paragraph{Colorful expansion.} Our method in asserting the rapid convergence of the high order random walks is as follows. We introduce a new notion of high order combinatorial expansion of complexes which we term colorful expansion. A $d$-dimensional simplicial complex is said to be a colorful expander if for any $0 \le i < d$ and any subset $S$ of $i$-faces there is a large set of {\em expanding} $(i+1)$-faces, i.e., faces which are hit by $S$ but are not fully covered by $S$. The term colorful expansion comes from the fact that if we color the faces in $S$ with one color and the faces not in $S$ with another color, then we get many colorful $(i+1)$-faces. This notion is of an independent interest as it is a natural generalization of combinatorial expansion of graphs to higher dimensions, which is not implied by previously studied notions such as coboundary or cosystolic expansion (see ~\cite{EK16}). We show that colorful expansion of a complex implies that {\em all} of its adjacency matrices $A_i$, $0 \le i < d$, have a concentrated spectrum, and hence all of the high order random walks on it converge rapidly to their stationary distribution.

We provide a local-to-global criterion on a complex which implies colorful expansion. For any face $\sigma$ in the complex, its {\em link} is the local view of $\sigma$, which is obtained by taking all the faces containing $\sigma$ and removing $\sigma$ from all of them. We prove that if the $1$-dimensional skeleton (i.e., the underlying graph) of every link in the complex 
is a spectral expander graph, then the entire complex is a colorful expander.

We conclude by showing that Ramanujan complexes satisfy this criterion, and hence are colorful expanders. We then use the explicit construction of Ramanujan complexes from~\cite{LSV05.2} in order to achieve an explicit family of bounded degree complexes in which all of the high order random walks converge rapidly to their stationary distribution.

\subsection{Expander graphs and random walks}\label{expanders-and-random-walks-section}
Let $G=(V,E)$ be an undirected graph. Denote by $A$ the adjacency matrix of $G$, where $A(u,v)$ equals the number of edges between $u$ and $v$ (we allow multiple edges). The {\em degree} of a vertex is the number of edges containing it. Throughout this section assume for simplicity that $G$ is $k$-regular, i.e., the degree of any vertex is $k$ (we deal with the non-regular case at Section~\ref{non-regular-section}). The {\em Cheeger constant} of $G$ is defined as
$$h(G) = \min_{\substack{\emptyset \ne S \subset V \\ |S| \le |V|/2}}\frac{|E(S,\bar S)|}{k|S|},$$
where $E(S,\bar S)$ is the set of edges with one endpoint in $S$ and one endpoint in $\bar S$.

If $h(G) \ge \epsilon$ for some constant $\epsilon > 0$, then $G$ is said to be an {\em $\epsilon$-combinatorial expander}. A family of graphs $\{G_i\}_{i \in \N}$ is called a family of $\epsilon$-combinatorial expanders if there exists a constant $\epsilon > 0$ such that $h(G_i) \ge \epsilon$ for every $i \in \N$. Intuitively it means that the graph is very well connected, and hence the random walk on it is not likely to ``get stuck'' in a small subset of vertices.

A {\em random walk} on $G$ can be described by the following process. We start at a random vertex $v_0 \in V$. Then, if after $t$ steps we are at vertex $v_t$, we choose one of the edges containing $v_t$ uniformly at random and move to its other endpoint. Thus, for any $u,v \in V$
$$\Pr[v_{t+1} = v \;|\; v_t = u] = \frac{A(u,v)}{k}.$$

Denote by $\pi_0 \in \R^V$ the probability distribution from which $v_0$ is chosen. Then the random walk on $G$ yields a sequence of probability distributions $\pi_0, \pi_1, \dotsc \in \R^V$. The transition from $\pi_t$ to $\pi_{t+1}$ is given by $\pi_{t+1} = \pi_t \widetilde A$, where $\widetilde A = (\rfrac{1}{k})A$ is the {\em normalized adjacency matrix} of $G$. Note that $\widetilde A(u,v)$ is the probability to move from $u$ to $v$ in a single step, and $\widetilde A^t(u,v)$ is the probability that a walk starting at $u$ will be at $v$ after $t$ steps. It follows that for any $t \in \N$, $\pi_t = \pi_0\widetilde A^t$.

Denote by $\mathbf u = (\rfrac{1}{|V|}, \dotsc, \rfrac{1}{|V|})$ the uniform distribution on the vertices. If $G$ is connected and non-bipartite, then $\pi_0\widetilde A^t \rightarrow \mathbf u$ as $t \rightarrow \infty$ for any initial probability distribution $\pi_0 \in \R^V$. The random walk on $G$ is said to be {\em $\mu$-rapidly mixing}, $0 < \mu < 1$, if for any initial probability distribution $\pi_0 \in \R^V$ and any $t \in \N$
$$\norm{\pi_t - \mathbf u}_2 \le \mu^t.$$

The common way to deduce the mixing rate of a random walk is by the spectrum of the normalized adjacency matrix of $G$. Denote by $1 = \lambda_1 \ge \lambda_2 \ge \dotsb \ge \lambda_{|V|} \ge -1$ the eigenvalues of $\widetilde A$. Then the random walk on $G$ is $\mu$-rapidly mixing for $\mu \le \max\{|\lambda_2|, |\lambda_{|V|}|\}$. (A proof for this assertion can be found in many places, see~\cite{HLW06} for instance.)

Here comes the relation between expanders and rapid mixing of random walks. Since combinatorial expanders are very strongly connected, we expect that random walks on them would mix rapidly. This is indeed true and it is given formally by Cheeger's inequality~\cite{Alo86} which states that
$$\lambda_2 \le 1 - \frac{h(G)^2}{2}.$$
Thus, it suffices to show that a graph is an $\epsilon$-combinatorial expander for some constant $\epsilon > 0$ in order to deduce that the random walk on it converges rapidly to the uniform distribution.

\begin{rem}\label{bound-lambda-n-rem}
	The cautious reader would note that combinatorial expansion gives us a bound only on $\lambda_2$, but for rapid mixing we need to bound $\lambda_{|V|}$ as well. This is not a problem since $\lambda_{|V|} = -1$ if and only if the graph is bipartite. Thus, in most applications it is possible just to add self-loops at each vertex which ``destroys the bipartiteness'' of the graph and does not slow down the convergence rate too much. Nevertheless, we show at section~\ref{trevisan-section} that when dealing with high order random walks the smallest eigenvalue of the normalized adjacency matrix is bounded away from $-1$. In any case, combinatorial expansion is the crucial condition for rapid mixing of random walks.
\end{rem}

\subsection{High dimensional simplicial complexes}
A {\em simplicial complex} $X$ with a set of vertices $V$ is a collection of subsets of $V$ with a closure property, namely, for any subset $\sigma \in X$, all of its subsets $\tau \subset \sigma$ are also in $X$. Each such subset is called a {\em face} of the complex. The {\em dimension} of a face $\sigma \in X$ is one less than the number of vertices in it, i.e., $\dim(\sigma)=|\sigma|-1$.
We also define the empty set as a face of the complex, which is the unique face of dimension $-1$. The dimension of the entire complex is defined as the maximal dimension of a face in it. The complex is said to be {\em pure} if for any face $\sigma \in X$ with $\dim(\sigma) < \dim(X)$ there is a face of maximal dimension $\tau \in X$, $\dim(\tau)=\dim(X)$, such that $\sigma \subset \tau$. We only deal with pure simplicial complexes.

We denote by $X(i)$ the faces of the complex of dimension $i$, which are called in short the $i$-faces of the complex. In the terminology of simplicial complexes, an {\em $i$-cochain} is a function $W:X(i) \rightarrow \{0,1\}$. For us it is convenient to view an $i$-cochain as defining a subset of $i$-faces, where $\sigma \in W \Leftrightarrow W(\sigma) = 1$. The space of all the $i$-cochains is denoted by $C^i(X)$.

The {\em degree} of a face $\sigma \in X$ is the number of faces of maximal dimension which contain $\sigma$. Again, we assume for simplicity that the complex is regular, i.e., for any $0 \le i < d$ and any two $i$-faces $\sigma, \tau \in X(i)$ it holds that $\deg(\sigma)=\deg(\tau)$ (we deal with non-regularity at Section~\ref{non-regular-section}). For any $i$-cochain $W \in C^i(X)$ we define its norm to be the fraction of $i$-faces which are in $W$ (this definition will be refined at Section~\ref{non-regular-section}), i.e.,
$$\norm{W} = \frac{|W|}{|X(i)|}.$$

The {\em link} of a face $\sigma \in X$, denoted by $X_\sigma$, is the complex obtained by taking all the faces which contain $\sigma$, and removing $\sigma$ from all of them. Intuitively, it is the local view of $X$ ``from the eyes of $\sigma$''. Formally it is given by $X_\sigma = \{\tau \setminus \sigma \;|\; \tau \in X,\, \tau \supseteq \sigma\}$. If $X$ is of dimension $d$, then $X_\sigma$ is a complex of dimension $d-|\sigma|$. We denote by $\norm{\,\cdot\,}_\sigma$ the norm associated with the link of $\sigma$.

Let us demonstrate the above definitions with a simple example.
\begin{exm}
	Let $X$ be a $2$-dimensional simplicial complex. The complex contains faces of dimensions $0,1$ and $2$, which can be viewed as vertices, edges and triangles. $X(0)$ denotes the vertices, $X(1)$ the edges, and $X(2)$ the triangles. The degree of any vertex $u \in X(0)$ is the number of triangles containing it, i.e., $\deg(u) = |\{t \in X(2) \;|\; u \in t\}|$. The link of a vertex $u \in X(0)$ is a $1$-dimensional simplicial complex defined as follows. Any edge of the form $\{u, v\} \in X(1)$ becomes a vertex in the link, i.e., $v \in X_u(0)$, and any triangle of the form $\{u, v, w\} \in X(2)$ becomes an edge in the link, i.e., $\{v, w\} \in X_u(1)$. Thus, the link of $u$ contains vertices and edges, which correspond to edges and triangles of the complex, respectively.
\end{exm}

\subsection{Our contribution}
In this work we introduce a new notion of random walks on high dimensional simplicial complexes. Our main result is a local criterion on a complex which implies the rapid mixing of all of its associated high order random walks. In order to assert their rapid mixing we introduce a new definition for combinatorial expansion of high dimensional simplicial complexes, which we term colorful expansion. We prove that this local criterion on a complex implies colorful expansion, which in turn implies the rapid convergence of all of the high order random walks on it. We then show an explicit construction of bounded degree complexes which satisfy this criterion, and hence are bounded degree complexes in which the high order random walks converge rapidly to their stationary distribution.

\subsubsection{High order random walks}
Let $X$ be a $d$-dimensional simplicial complex. We generalize the notion of random walks on graphs to higher dimensions by letting the walk to move on any dimension of the complex. For any $0 \le i < d$ we define the walk on dimension $i$ of the complex as follows. We start from an initial $i$-face $\sigma_0 \in X(i)$. Then, if after $t$ steps we are at $\sigma_t$, the process of choosing $\sigma_{t+1}$ can be described by the following two steps:
\begin{enumerate}
  \item Choose an $(i+1)$-face $\tau \supset \sigma_t$ uniformly at random.
  \item Choose an $i$-face $\sigma_{t+1} \subset \tau$, $\sigma_{t+1} \ne \sigma_t$, uniformly at random and move to it.
\end{enumerate}

We say that $\sigma, \sigma' \in X(i)$ are neighbors, denoted by $\sigma \sim \sigma'$, if they share a common $(i+1)$-face, i.e., $\sigma \cup \sigma' \in X(i+1)$. In order to analyze this walk, we define the following auxiliary graphs.
\begin{dfn}
	Let $X$ be a $d$-dimensional simplicial complex. For any $0 \le i < d$ the {\em $i$-graph} $G_i = (V_i,E_i)$ (with adjacency matrix $A_i$) is defined as follows.
	\begin{itemize}
		\item The vertices of $G_i$ are the $i$-faces of $X$, i.e., $V_i = X(i)$.
		\item For any two $i$-faces $\sigma,\sigma' \in X(i)$ such that $\sigma \sim \sigma'$, there is an edge between the corresponding vertices in $G_i$, i.e., $E_i=\left\{\!\{\sigma, \sigma'\} \in \binom{X(i)}{2} \mid \sigma \sim \sigma'\right\}$.
	\end{itemize}
\end{dfn}

Denote by $\pi_0 \in \R^{X(i)}$ the probability distribution from which $\sigma_0$ is chosen. As in graphs, the random walk on dimension $i$ of $X$ yields a sequence of probability distributions $\pi_0, \pi_1, \dotsc, \in \R^{X(i)}$. Note that the transition from $\pi_t$ to $\pi_{t+1}$ is given by $\pi_{t+1} = \pi_t \widetilde A_i$, where $\widetilde A_i$ is the normalized adjacency matrix of $G_i$. In other words, the random walk on dimension $i$ of the complex is equivalent to a random walk on the $i$-graph as above. Thus, our goal is to construct a complex $X$ such that the random walks on {\em all} of its induced $i$-graphs, $0 \le i < d$, will converge rapidly to the uniform distribution.

\subsubsection{Rapid mixing criterion}
Recall that the link of any face $\sigma \in X$ is a smaller complex which corresponds to the local view of $\sigma$. We essentially show that if every link is expanding in some sense, then all of the high order random walks on $X$ converge rapidly to their stationary distribution.

For any face $\sigma \in X$ the $1$-dimensional skeleton of $X_\sigma$ is defined as $X_\sigma(0) \cup X_\sigma(1)$, i.e., the underlying graph of $X_\sigma$. We say that $X$ is a {\em skeleton expander} if the $1$-dimensional skeleton of every link behaves pseudorandomly.
\begin{dfn}[Skeleton expansion]
	Let $X$ be a $d$-dimensional simplicial complex. $X$ is called an {\em $\alpha$-skeleton expander}, $\alpha > 0$, if for any face $\sigma \in X$ (including $\sigma = \emptyset)$ and any subset of vertices $S \subseteq X_\sigma(0)$ it holds that
	$$\norm{E(S)}_\sigma \le \norm{S}_\sigma(\norm{S}_\sigma + \alpha),$$
	where $E(S) \subseteq X_\sigma(1)$ denotes the set of edges with both endpoints in $S$.
\end{dfn}

We note that in the above definition we can replace the requirement of $\alpha$-skeleton expansion for $\sigma = \emptyset$ with the weaker requirement that the complex is connected. This could be done since by~\cite{Opp15} if the complex is connected, then the $\alpha$-skeleton expansion of  $\sigma =\emptyset$ follows from the $\alpha$-skeleton expansion of all $\sigma \in X$, $\sigma \neq \emptyset$.

For some intuition regrading the above definition, consider a random graph $G = (V,E)$ whose edges are distributed uniformly. Let $S \subseteq V$ be any subset of vertices of size $\gamma|V|$. Then the probability for any edge to have both endpoints in $S$ is $\gamma^2$, and hence the expected size of $E(S)$ is $\gamma^2|E|$. From this point of view, the $\alpha$-skeleton expansion property means that the underlying graph of every link ``looks like'' a random graph, up to an error of $\alpha$.

The main result of this paper is the following theorem.
\begin{thm}[Rapid mixing criterion, informal, for formal see Theorem \ref{main-thm-3}]\label{main-thm-3-informal}
	If a complex is an $\alpha$-skeleton expander for small enough $\alpha$, then all of the high order random walks on it converge rapidly to their stationary distribution.
\end{thm}

\subsubsection{Colorful expansion}
Our new definition for combinatorial expansion of high dimensional simplicial complexes is the main building block in the rapid mixing proof. For an $i$-cochain $W \in C^i(X)$, we say that an $(i+1)$-face is {\em expanding} if it hits $W$ but not fully covered by $W$.
\begin{dfn}[Expanding faces]
	Let $X$ be a $d$-dimensional simplicial complex. For any $i$-cochain $W \in C^i(X)$, $0 \le i < d$, the {\em expanding faces} of $W$ are defined as
	$$\psi(W) = \{\tau \in X(i+1) \;|\; \exists\, \sigma, \sigma' \subset \tau: \sigma \in W,\, \sigma' \in X(i) \setminus W\}.$$
\end{dfn}

Then we define colorful expansion as follows.
\begin{dfn}[Colorful expander]
	Let $X$ be a $d$-dimensional simplicial complex. $X$ is called an {\em $\epsilon$-colorful expander}, $\epsilon > 0$, if for any $i$-cochain $W \in C^i(X)$, $0 \le i < d$, $0 < \norm{W} \le 1/2$,
	$$\frac{\norm{\psi(W)}}{\norm{W}} \ge \epsilon.$$
\end{dfn}

Our main result is then a corollary of the following two theorems.
\begin{thm}[Colorful expansion criterion, informal, for formal see Theorem \ref{main-thm-1}]\label{main-thm-1-informal}
	If a complex is an $\alpha$-skeleton expander for small enough $\alpha$, then it is a colorful expander.
\end{thm}

\begin{thm}[Colorful expansion implies rapid mixing, informal, for formal see Theorem \ref{main-thm-2}]\label{main-thm-2-informal}
	If a complex is a colorful expander, then all of the high order random walks on it converge rapidly to their stationary distribution.
\end{thm}

Since Theorem~\ref{main-thm-1-informal} is the main technical part of this paper, we provide here a short overview of the proof.

\paragraph{Informal overview of the proof of Theorem~\ref{main-thm-1-informal}.} We need to show that for any dimension $i$ and any non-empty $i$-cochain with norm up to $1/2$ there are many expanding faces.

Consider an arbitrary $i$-cochain $W$. From a very high level point of view, the proof can be divided into three steps:
\begin{enumerate}
	\item Mark ``selected'' faces in dimensions $i,i-1,\dotsc,-1$.
	\item Find a ``good'' dimension $j \le i$.
	\item ``Climb up'' from dimension $j$ to deduce many expanding faces in dimension $i+1$.
\end{enumerate}

In more details, we start by marking all of the $i$-faces which are in $W$. Then we mark $(i-1)$-faces, where the rule for marking a face is if many (above some threshold) of the $i$-faces which contain it are marked. We continue this way down to the lowest dimension of the complex, where each dimension is marked with regard to one dimension above. We term these marked faces as {\em fat faces} since each marked face contains in its link many $i$-faces which are in $W$.

The next step is to look for the highest dimension $j \le i$ with the property that many fat $j$-faces contain many non-fat $(j-1)$-faces. By setting a large enough threshold for marking faces we can make sure that the empty set is never fat, which implies that such dimension must exist (Proposition~\ref{good-dimension-pro}).

The last step is to deduce from the links of the non-fat $(j-1)$-faces that there exist many expanding $(i+1)$-faces. We already know, by the existence of many fat $j$-faces in these links, that a large part of $W$ is seen in them (Proposition~\ref{container-bound-pro}). It is left to show that many of the $(i+1)$-faces in these links contain at least one $i$-face which is not in $W$ (Proposition~\ref{full-faces-upper-bound-pro}).

For that we define {\em full faces} for any dimension of the complex, where a $k$-face is full if it contains only fat $(k-1)$-faces. By this definition it follows that a full $k$-face is seen in the link of any $(k-2)$-face as an edge between two fat vertices. So when considering the link of a non-fat $(k-2)$-face, which by definition contains only a few fat vertices, by the skeleton expansion  it contains only a few full $k$-faces. Thus, the number of the full $(j+1)$-faces in the links of the non-fat $(j-1)$-faces is very small. Then again, by the skeleton expansion and the maximality of $j$ we can deduce that for any $k>j+1$ most of the full $k$-faces contain only full $(k-1)$-faces. This implies that the fraction of full $(i+1)$-faces is about the same as the fraction of the full $(j+1)$-faces. This actually finishes the proof, since we get many $(i+1)$-faces that are not full, i.e., contain at least one $i$-face which is not in $W$.

\subsubsection{Explicit construction}
As proven in~\cite{EK16}, Ramanujan complexes are excellent skeleton expanders. Since our definition of skeleton expansion is slightly different than the one appearing in~\cite{EK16}, we modify their proof so the requirement of our definition follows. We show that Ramanujan complexes with thickness large enough are $\alpha$-skeleton expanders with $\alpha$ as small as we want. Then by the explicit construction of Ramanujan complexes from~\cite{LSV05.2} we achieve the following theorem.
\begin{thm}[Explicit construction, informal, for formal see Theorem \ref{main-thm-4}]\label{main-thm-4-informal}
	There exists an explicit family of bounded degree high dimensional simplicial complexes in which all of the high order random walks converge rapidly to their stationary distribution.
\end{thm}

\subsection{Related work}
In a recent work Oppenheim~\cite{Opp15} has shown that if all the links of a $d$-dimensional simplicial complex are good spectral expanders and the complex is connected, then the $1$-dimensional skeleton of the complex is an expander graph, or in other words, that the graph $G_0$ is a spectral expander. In this paper we show that the spectral expansion of the links, including the spectral expansion of the $1$-dimensional skeleton of the complex, imply the spectral expansion of $G_i$ for {\em all} $0 < i < d$.

Parzanchevski and Rosenthal in~\cite{PR12} study a different notion of high order random walk. The high order random walk of~\cite{PR12} is designed to expose the topological properties of the complex. Parzanchevski and Rosenthal define a variant of the stationary distribution of the random walk and show its relation to the spectrum of the high order laplacian on the space that is {\em orthogonal to the coboundaries}. In our work, the stationary distribution of the random walks is already known. Moreover, it is already known that the convergence rate is controlled by the spectrum of the high order normalized adjacency matrices on the space that is {\em orthogonal to the constant functions}. We show that the expansion of the links implies the concentration of the spectrum of the high order normalized adjacency matrices, a thing that can not be deduced in any way from~\cite{PR12}.

We study the spectrum of the high order normalized adjacency matrices of the complex and derive some good bounds on it from the spectrum of the links. Garland in a seminal work~\cite{Gar73} has studied the spectrum of high order laplacians associated with a $d$-dimensional simplicial complex. Garland has shown that if all the links of a complex are good spectral expanders, then the eigenspace of the weighted oriented laplacian that is orthogonal to the coboundaries has a concentrated spectrum. This is somewhat in the spirit of what we get here. However, Garland could only obtain bounds on the eigenspace orthogonal to the coboundaries, while here we want to get a bound on all the eigenvalues besides the first trivial one corresponding to the constant functions. There is no known way to obtain our result from Garland's argument. See also~\cite{GW12} for more discussion on Garland's work and the fact that it does not imply the required spectrum bound for the eigenvalues on the space that is orthogonal to the constant functions.

In a recent work of the first coauthor and Evra~\cite{EK16} a different notion of high order expansion has been studied, which is called cosystolic expansion. It was shown in~\cite{EK16} that if all the links of a $d$-dimensional simplicial complex are good spectral expanders and good coboundary expanders, then the $(d-1)$-skeleton of the complex is a cosystolic expander. Though the spirit of the proof there might resemble at first glance the method of the proof that we use here, the obstacles and the solutions are different. They could only show cosystolic expansion of {\em small} sets. Then they use a reduction of~\cite{KKL14} showing that cosystolic expansion of small sets implies cosystolic expansion of the $(d-1)$-skeleton of a given $d$-dimensional simplicial complex. In our work it is crucial for us to obtain expansion of large sets, whose norm is up to $1/2$, and the reduction of~\cite{KKL14} could not work here since it does not imply colorful expansion. Thus, the expansion here is achieved in a method which is different than the one used in~\cite{EK16}.

\subsection{Dealing with non-regularity}\label{non-regular-section}
Up to now we assumed that the complexes are regular. In this section we describe the necessary modifications for the general case, where we are not guaranteed to have regularity. We start with non-regular graphs (for more about random walks on non-regular graphs see~\cite{SJ89}).

Let $G=(V,E)$ be an undirected graph. Instead of using the Cheeger constant we use the {\em conductance}, which is its generalized version. For any subset of vertices $S \subseteq V$, its {\em volume} is defined as $\vol(S) = \sum_{v \in S}\deg(v)$. Then the conductance of any subset of vertices $S \subseteq V$ is defined as
$$\Phi(S) = \frac{|E(S,\bar S)|}{\vol(S)},$$
and the conductance of the graph is defined as
$$\Phi(G) = \min_{\substack{\emptyset \ne S \subseteq V \\ \vol(S) \le \vol(V)/2}}\Phi(S).$$

When considering a random walk on $G$, then for any $u,v \in V$ and any $t \in \N$
$$\Pr[v_{t+1} = v \;|\; v_t = u] = \frac{A(u,v)}{\deg(u)}.$$
So the transition from $\pi_t$ to $\pi_{t+1}$ is given by $\pi_{t+1} = \pi_t(DA)$, where $D$ is the diagonal matrix defined by
\[
D(u,v) =
\begin{cases}
\displaystyle\frac{1}{\deg(u)} & \text{if } u=v, \\
0 & \text{otherwise}.
\end{cases}
\]

The {\em stationary distribution} of the walk is the probability distribution $\pi \in \R^V$ for which $\pi (DA) = \pi$. When $G$ is connected and non-bipartite, then for any $v \in V$
$$\pi(v) = \frac{\deg(v)}{\sum_{u \in V}\deg(u)},$$
and for any initial probability distribution $\pi_0 \in \R^V$ it holds that $\pi_0(DA)^t \rightarrow \pi$ as $t \rightarrow \infty$. The random walk is said to be {\em $\mu$-rapidly mixing}, $0 < \mu < 1$, if for any initial probability distribution $\pi_0 \in \R^V$ and any $t \in \N$
$$\norm{\pi_t - \pi}_2 \le \sqrt{\frac{d_{max}}{d_{min}}}\mu^t,$$
where $d_{max}=\max_{v \in V}\{\deg(v)\}$ and $d_{min}=\min_{v \in V}\{\deg(v)\}$.

The mixing rate of the random walk on $G$ can be deduced from the spectrum of its normalized adjacency matrix, which is defined in the general case as $\widetilde A=D^{\rfrac{1}{2}}AD^{\rfrac{1}{2}}$. The largest eigenvalue of $\widetilde A$ satisfies $\lambda_1 = 1$ (with corresponding eigenvector $\big(\!\sqrt{\pi(v)}\,\big)_{v \in V}$) and the smallest eigenvalue satisfies $\lambda_{|V|} \ge -1$. Similar to the regular case, the random walk on $G$ is $\mu$-rapidly mixing for $\mu \le \max\{|\lambda_2|, |\lambda_{|V|}|\}$. (We prove this assertion in the appendix.)

The following lemma of Sinclair and Jerrum~\cite{SJ89} generalizes Cheeger's inequality to non-regular graphs.
\begin{lem}\label{cheeger-lambda-2}
	Let $G=(V,E)$ be an undirected graph where multiple edges are allowed, $\widetilde A$ its normalized adjacency matrix and $\lambda_1 \ge \lambda_2 \ge \dotsb \ge \lambda_{|V|}$ the eigenvalues of $\widetilde A$. Then
	$$\lambda_2 \le 1 - \frac{\Phi(G)^2}{2}.$$
\end{lem}

We describe now the modifications required for non-regular complexes. Let $X$ be a $d$-dimensional simplicial complex. For any $i$-cochain $W \in C^i(X)$ we refine the definition of the norm to the general case as
$$\norm{W} = \frac{\sum_{\sigma \in W}deg(\sigma)}{\sum_{\tau \in X(i)}deg(\tau)}.$$

An alternative view of this norm, which is very useful for further calculations, is the following. Let $P_d \in X(d)$ be a uniformly random $d$-face of the complex. For any $i = d-1,\dotsc, -1$, let $P_i \in X(i)$ be a random $i$-face which is obtained by removing a uniformly random vertex from $P_{i+1}$. Then we get a sequence of random variables $\{P_i\}_{i=-1}^d$ such that for any $i$-cochain $W \in C^i(X)$
$$\norm{W} = \Pr[P_i \in W].$$

Now when considering a random walk on dimension $i$ of $X$, we take into consideration the degrees of the faces. Assuming the walk is now at $\sigma_t$, then the process of choosing $\sigma_{t+1}$ is given as follows:
\begin{enumerate}
	\item Choose an $(i+1)$-face $\tau \supset \sigma_t$ with probability {\em proportional to its degree}.
	\item Choose an $i$-face $\sigma_{t+1} \subset \tau$, $\sigma_{t+1} \ne \sigma_t$, uniformly at random and move to it.
\end{enumerate}
So the exact probability for moving from $\sigma$ to $\sigma'$, where $\sigma \sim \sigma'$, is given by
$$\Pr[\sigma_{t+1} = \sigma' \;|\; \sigma_t = \sigma] = \frac{\deg(\sigma \cup \sigma')}{\sum_{\sigma'' \sim \sigma}\deg(\sigma \cup \sigma'')}.$$

The last generalization required is for the auxiliary graphs $G_i$, $0 \le i < d$. For any two $i$-faces $\sigma,\sigma' \in X(i)$ such that $\sigma \sim \sigma'$, instead of having one edge between the corresponding vertices in $G_i$, we put $\deg(\sigma \cup \sigma')$ edges between them. Now the random walk on $G_i$ is equivalent to the random walk on dimension $i$ of $X$, i.e., for any $t \in \N$ it holds that $\pi_{t+1} = \pi_t(D_iA_i)$, where $\pi_t$ is the probability distribution after $t$ steps of the random walk on dimension $i$ of $X$, and $D_i,A_i$ are the diagonal and adjacency matrices of $G_i$, respectively.

\subsection{Bounding the smallest eigenvalue}\label{trevisan-section}
The last technical issue, mentioned in remark~\ref{bound-lambda-n-rem}, is to make sure that $\lambda_{|V|}$ is bounded away from $-1$. For random walks on dimension $i \ge 1$ we can do that by another combinatorial measure of the graph, which we describe in this section.

Let $G=(V,E)$ be an undirected graph. A {\em bipartite component} in $G$ is a non-empty subset of vertices $\emptyset \ne S \subseteq V$ and a partition of $S$ into two disjoint subsets $S_1 \cupdot S_2 = S$ such that $|E(S, \bar S)| = |E(S_1)| = |E(S_2)| = 0$. ($E(S_i)$ denotes the set of edges with both endpoints in $S_i$.) It is known that $\lambda_{|V|} = -1$ if and only if the graph has a bipartite component. Trevisan~\cite{Tre16} has proven that when the graph is ``far'' from having a bipartite component, then $\lambda_{|V|}$ is bounded away from $-1$.

For any tuple $(S, S_1, S_2)$, $\emptyset \ne S \subseteq V$, $S_1 \cupdot S_2 = S$, the {\em bipartiteness ratio} of $(S, S_1, S_2)$ is defined as
$$\beta(S, S_1, S_2) = \frac{|E(S,\bar S)| + 2(|E(S_1)| + |E(S_2)|)}{\vol(S)},$$
and the bipartiteness ratio of the graph is defined as
$$\beta(G) = \min_{\substack{\emptyset \ne S \subseteq V \\ S_1 \cupdot S_2 = S}}\beta(S, S_1, S_2).$$

It is easy to see that $\beta(G) = 0$ if and only if $G$ has a bipartite component. Moreover, when $\beta(G)$ is close to $0$, then $G$ has a component which is ``almost'' bipartite. As $\beta(G)$ is bounded away from $0$, $G$ is further away from having a bipartite component. The following lemma of Trevisan~\cite{Tre16} formalizes it.
\begin{lem}\label{cheeger-lambda-n}
	Let $G=(V,E)$ be an undirected graph where multiple edges are allowed, $\widetilde A$ its normalized adjacency matrix and $\lambda_1 \ge \lambda_2 \ge \dotsb \ge \lambda_{|V|}$ the eigenvalues of $\widetilde A$. Then
	$$\lambda_{|V|} \ge -1 + \frac{\beta(G)^2}{2}.$$
\end{lem}

When dealing with a random walk dimension $0$, i.e., the vertices of the complex, we cannot guarantee that the graph is far from being bipartite. In this case we add $\deg(v)$ self-loops to any vertex $v \in V_0$. Then we get a {\em lazy} random walk on the vertices of the complex, where at each step we stay put with probability $1/2$. It turns out that this technique cancels the negative eigenvalues of the complex. The following proposition follows from elementary linear algebra and can also be found at~\cite{SJ89}.
\begin{pro}\label{self-loops-prop}
	Let $G=(V,E)$ be an undirected graph where multiple edges are allowed, $\widetilde A$ its normalized adjacency matrix and $\lambda_1 \ge \lambda_2 \ge \dotsb \ge \lambda_{|V|}$ the eigenvalues of $\widetilde A$. Denote by $G'=(V,E')$ the modified graph after adding $\deg(v)$ self-loops to any vertex $v \in V$. Then the eigenvalues $\lambda'_2$ and $\lambda'_{|V|}$ of the modified normalized adjacency matrix $\widetilde A'$ satisfy $\lambda'_2 = (1+\lambda_2)/2$ and $\lambda'_{|V|} \ge 0$.
\end{pro}

\subsection{Organization of this paper}
The paper is organized as follows. In Section $2$ we prove that the expansion of the links imply colorful expansion (Theorem~\ref{main-thm-1-informal}). In Section $3$ we prove the rapid mixing criterion for the high order random walks (Theorems~\ref{main-thm-2-informal} and~\ref{main-thm-3-informal}). In Section $4$ we show that Ramanujan complexes satisfy this criterion and hence form an explicit family of bounded degree high dimensional complexes in which all of the high order random walks mix rapidly (Theorem~\ref{main-thm-4-informal}).

\subsection{Acknowledgements}
We would like to thank Roei Tell for introducing us to the probabilistic view of the norm, it made the proofs much simpler.

\section{Colorful expansion}
In this section we prove that the expansion of all the links of a complex imply its colorful expansion. The following is the formal version of Theorem \ref{main-thm-1-informal} from the introduction.
\begin{thm}[Colorful expansion criterion]\label{main-thm-1}
	Let $X$ be a $d$-dimensional $\alpha$-skeleton expander. If $\alpha < (\sqrt[2^d]{2}-1)/\sqrt{2}$, then there exists $\epsilon = \epsilon(d,\alpha)$ such that $X$ is an $\epsilon$-colorful expander. Specifically, it holds for
	$$\epsilon = \left(\frac{\sqrt[2^d]{2}-1-\sqrt2\alpha}{2\sqrt2d}\right)^d.$$
\end{thm}

This section is organized as follows. In Section~$2.1$ we define the required ingredients of the proof in a formal way. In Section~$2.2$ we just state the three propositions which are the main building blocks in the proof of the theorem. In Section~$2.3$ we prove the theorem assuming these propositions hold, and then in Section~$2.4$ we prove the main propositions which we stated earlier.

\subsection{Definitions}
Since we prove a global expansion property from the local expansion of the links, it is useful for us to relate cochains in the complex to their local views in the links.
\begin{dfn}[Localization]
	Let $X$ be a $d$-dimensional simplicial complex. For any $i$-cochain $W \in C^i(X)$, $0 \le i \le d$, and any $j$-face $\sigma \in X(j)$, $j < i$, the {\em localization} of $W$ to the link of $\sigma$ is defined as
	$$W_\sigma = \{\tau \in X_\sigma \;|\; \tau \cupdot \sigma \in W\}.$$
\end{dfn}

Recall that for any $i$-cochain $W \in C^i(X)$ it holds that $\norm{W} = \Pr[P_i \in W]$. When considering the localization of $W$ to a link of some $j$-face $\sigma \in X(j)$ ($j < i$) we have
\begin{equation}\nonumber
\norm{W_\sigma}_\sigma = \Pr[P_i \in W \;|\; P_j = \sigma].
\end{equation}
This holds because the condition $P_j = \sigma$ implies that $P_d$ is uniformly distributed on the $d$-faces containing $\sigma$, which are the maximal faces in the link of $\sigma$.

For any $i$-cochain $W \in C^i(X)$ and any dimension $j \le i$ we define the $j$-cochain of fat faces. These faces are defined recursively such that any $i$-face in $W$ is fat, and any $j$-face, $j<i$, is fat if many of the $(j+1)$-faces which contain it are fat.
\begin{dfn}[Fat faces]
	Let $X$ be a $d$-dimensional simplicial complex, $0 \le i \le d$, $W \in C^i(X)$ an arbitrary $i$-cochain and $0 < \eta < 1$ a fatness constant. The $i$-cochain of fat faces is defined as $S^i(W) = W$, and for any $-1 \le j < i$
	$$S^j(W) = \left\{\sigma \in X(j) \;|\; \norm{S^{j+1}(W)_\sigma}_\sigma \ge \eta^{2^{i-j-1}}\right\}.$$
\end{dfn}

Recall that the expanding faces of any $i$-cochain $W \in C^i(X)$ are the $(i+1)$-faces which contain at least one face in $W$ and one face not in $W$. For a simple representation of the expanding faces we define the following cochains.
\begin{dfn}[Container]
	Let $X$ be a $d$-dimensional simplicial complex.
	For any $i$-cochain $W \in C^i(X)$, $0 \le i \le d$, the {\em container} of $W$ is defined as
	$$\Gamma(W) = \{\tau \in X(i+1) \;|\; \exists \, \sigma \subset \tau: \sigma \in W\}.$$
\end{dfn}

\begin{dfn}[Full faces]
	Let $X$ be a $d$-dimensional simplicial complex.
	For any $i$-cochain $W \in C^i(X)$, $0 \le i < d$, the {\em full faces} of dimension $j$, $0 \le j \le i+1$, are defined as
	$$F^j(W) = \{\tau \in X(j) \;|\; \forall \, \sigma \subset \tau,\; \dim(\sigma)=j-1: \sigma \in S^{j-1}(W)\}.$$
\end{dfn}

In words, the container of $W$ is the $(i+1)$-cochain of faces which contain at least one $i$-face in $W$. The full faces are defined for any dimension $j \le i+1$ and they are the $j$-faces which all of their $(j-1)$-faces are fat. Then we get
$$\psi(W) = \Gamma(W) \setminus F^{i+1}(W).$$

\subsection{Main propositions}
The first proposition ensures that, with the right parameters, there exists some dimension $j \le i$ in which many fat $j$-faces contain many non-fat $(j-1)$-faces.
\begin{pro}[Existence of dimension in which many fat faces contain many non-fat faces]\label{good-dimension-pro}
	Let $X$ be a $d$-dimensional simplicial complex, $0 \le i < d$, $0 < \eta < 1$ a fatness constant.
	Then for any $i$-cochain $W \in C^i(X)$, $\norm{W} < \eta^{2^{i+1}-1} $, and any $c \le 1/2$ there exists $0 \le j \le i$ such that
	$$\Pr\!\big[P_j \in S^j(W) \wedge P_{j-1} \notin S^{j-1}(W)\big] \ge \frac{c^j}{i+1}\norm{W}.$$
\end{pro}

We relate to these non-fat $(j-1)$-faces as faces with a large ``potential'' since they have many fat $j$-faces in their links. We show that this potential can be ``lifted'' up to dimension $i+1$ in order to deduce that many faces of $W$ contain these non-fat $(j-1)$-faces.
\begin{pro}[Large norm of fat faces implies large norm of the container of $W$]\label{container-bound-pro}
	Let $X$ be a $d$-dimensional simplicial complex, $0 \le i < d$, $0 < \eta < 1$ a fatness constant. Then for any $i$-cochain $W \in C^i(X)$ and any $0 \le j \le i$
	$$\Pr[P_{i+1} \in \Gamma(W) \wedge P_{j-1} \notin S^{j-1}(W)] \ge \eta^{2^{i-j}-1}\Pr[P_j \in S^j(W) \wedge P_{j-1} \notin S^{j-1}(W)].$$
\end{pro}

The last proposition shows that, in a skeleton expander complex, the non-fat $(j-1)$-faces can be ``lifted'' up to dimension $i+1$ in order to deduce that many of the $(i+1)$-faces, which contain these non-fat $(j-1)$-faces, are not full.
\begin{pro}[Large norm of non-fat faces implies large norm of non-full faces]\label{full-faces-upper-bound-pro}
	Let $X$ be an $\alpha$-skeleton expander, $0 \le i < d$, $0 < \eta < 1$ a fatness constant. Then for any $i$-cochain $W \in C^i(X)$ and any $0 \le j \le i$
	\begin{displaymath}
	\begin{split}
	\Pr[P_{i+1} \in\, &F^{i+1}(W) \wedge P_{j-1} \notin S^{j-1}(W)] \le
	\\[6pt]&\big(\eta^{2^{i-j}} + \alpha\big)\Pr[P_j \in S^j(W) \wedge P_{j-1} \notin S^{j-1}(W)]\; + \\&\sum_{k=j+1}^i(k+1)\big(\eta^{2^{i-k}} + \alpha\big)\Pr[P_k \in S^k(W) \wedge P_{k-1} \notin S^{k-1}(W)]
	\end{split}
	\end{displaymath}
\end{pro}

\subsection{Proof of Theorem \ref{main-thm-1}}
Recall that for colorful expansion we need that any $i$-cochain with norm up to $1/2$ (of any dimension $i<d$) will have many expanding faces. Let $W \in C^i(X)$, $0 \le i < d$, be an arbitrary $i$-cochain with $\norm{W} \le 1/2$. Set $\eta = \sqrt[2^{i+1}]{1/2}$ and $$c = \frac{(\eta^{-1}-1)\eta^{2^i} - \alpha}{2(i+1)}.$$

By Proposition \ref{good-dimension-pro} there exists a dimension $0 \le j \le i$ in which many fat $j$-faces contain many non-fat $(j-1)$-faces, i.e.,
\begin{equation}\label{main-theorem-eq-1}
\Pr[P_j \in S^j(W) \wedge P_{j-1} \notin S^{j-1}(W)] \ge \frac{c^j}{i+1}\norm{W}.
\end{equation}
Let $j$ be the maximal which satisfies~\eqref{main-theorem-eq-1}, so for any $k \in \{j+1,\dotsc,i\}$
\begin{equation}\label{main-theorem-eq-2}
\Pr[P_k \in S^k(W) \wedge P_{k-1} \notin S^{k-1}(W)] < \frac{c^k}{i+1}\norm{W}.
\end{equation}

We are going to show that the non-fat $(j-1)$-faces imply many expanding $(i+1)$-faces. By definition we have
\begin{equation}\label{main-theorem-eq-3}
\begin{split}
\norm{\psi(W)} &=
\norm{\Gamma(W) \setminus F^{i+1}(W)} =
\Pr[P_{i+1} \in \Gamma(W) \setminus F^{i+1}(W)] \ge
\\[4pt]&\Pr[P_{i+1} \in \Gamma(W) \setminus F^{i+1}(W) \wedge P_{j-1} \notin S^{j-1}(W)] \ge
\\[4pt]&\Pr[P_{i+1} \in \Gamma(W) \wedge P_{j-1} \notin S^{j-1}(W)] - \Pr[P_{i+1} \in F^{i+1}(W) \wedge P_{j-1} \notin S^{j-1}(W)],
\end{split}
\end{equation}
where both of the inequalities follow from laws of probability. The meaning of this is that the expanding faces of $W$ are at least the difference between the container of $W$ and the full $(i+1)$-faces, when both are restricted to the non-fat $(j-1)$-faces.

By the existence of many fat $j$-faces in the links of the non-fat $(j-1)$-faces we know that a large part of $W$ is seen in these links, or in other words, we can ``lift'' the fat $j$-faces in order to deduce a lower bound on the container of $W$. This lower bound is given formally by Proposition~\ref{container-bound-pro}:
\begin{equation}\label{main-theorem-eq-4}
\Pr[P_{i+1} \in \Gamma(W) \wedge P_{j-1} \notin S^{j-1}(W)] \ge
\eta^{2^{i-j}-1}\Pr[P_j \in S^j(W) \wedge P_{j-1} \notin S^{j-1}(W)].
\end{equation}

By the skeleton expansion of $X$ we can ``lift'' the non-fat $(j-1)$-faces in order to deduce an upper bound on the full $(i+1)$-faces. So by Proposition~\ref{full-faces-upper-bound-pro} we have
\begin{equation}\label{main-theorem-eq-5}
\begin{split}
\Pr[P_{i+1} \in\, &F^{i+1}(W) \wedge P_{j-1} \notin S^{j-1}(W)] \le
\\[6pt]&(\eta^{2^{i-j}} + \alpha)\Pr[P_j \in S^j(W) \wedge P_{j-1} \notin S^{j-1}(W)]\; + \\&\sum_{k=j+1}^i(k+1)(\eta^{2^{i-k}} + \alpha)\Pr[P_k \in S^k(W) \wedge P_{k-1} \notin S^{k-1}(W)]
\end{split}
\end{equation}

Substituting~\eqref{main-theorem-eq-4} and~\eqref{main-theorem-eq-5} in~\eqref{main-theorem-eq-3} yields
\begin{equation}\label{main-theorem-eq-6}
\begin{split}
\norm{\psi(W)} \ge\:
&\big((\eta^{-1}-1)\eta^{2^{i-j}}-\alpha\big)\Pr[P_j \in S^j \wedge P_{j-1} \notin S^{j-1}(W)]\: - \\&\sum_{k=j+1}^i(k+1)(\eta^{2^{i-k}}+\alpha)\Pr[P_k \in S^k \wedge P_{k-1} \notin S^{k-1}(W)].
\end{split}
\end{equation}

Now by substituting~\eqref{main-theorem-eq-1} and~\eqref{main-theorem-eq-2} in~\eqref{main-theorem-eq-6} we get
\begin{equation}\nonumber
\begin{split}
\norm{\psi(W)} &\ge
\left(\frac{(\eta^{-1}-1)\eta^{2^{i-j}}-\alpha}{i+1}c^j - \sum_{k=j+1}^i\frac{k+1}{i+1}(\eta^{2^{i-k}}+\alpha)c^k\right)\norm{W} \\&\ge
c^j\left(\frac{(\eta^{-1}-1)\eta^{2^i}-\alpha}{i+1} - \sum_{k=1}^{i-j}\frac{k+j+1}{i+1}c^k\right)\norm{W} \ge
c^j(2c - c)\norm{W} =
c^{j+1}\norm{W},
\end{split}
\end{equation}
where the second inequality holds because
\begin{equation}\nonumber
\eta^{2^{i-k}} + \alpha \le
\eta + \alpha <
\sqrt[2^{i+1}]{1/2} + \frac{\sqrt[2^{i+1}]{2}-1}{\sqrt2} =
2^{-x} + \frac{2^x-1}{\sqrt2}
\end{equation}
for $x = (1/2)^{i+1}$. Since $k \ge 1$ then $0 < x \le 1/4$ and hence the function is decreasing and gets its maximum at $x=0$, which yields $\eta^{2^{i-k}}+\alpha < 1$ for any $k \ge 1$. The last inequality holds trivially for $j=i$, and for $j < i$ it holds since
\begin{equation}\nonumber
\begin{split}
\sum_{k=1}^{i-j}\frac{k+j+1}{i+1}c^k &\le
c\left(\frac{i+1-(i-j-1)}{i+1} + \sum_{k=2}^{i-j}c^{k-1}\right) \\&\le
c\left(\frac{i+1-(i-j-1)}{i+1} + (i-j-1)\frac{1}{i+1}\right) \le c.
\end{split}
\end{equation}

Since $j \le i < d$ it follows that for any $i$-cochain $W \in C^i(X)$, $1 \le i < d$, $\norm{W} \le 1/2$,
\begin{equation}\nonumber
\frac{\norm{\psi(W)}}{\norm{W}} \ge c^{i+1} \ge \epsilon
\end{equation}
for $\epsilon$ as in the theorem.
%
\qed

\subsection{Proofs of the main propositions}
For ease of notation, in this section we use the following shortcuts.
\begin{itemize}
	\item When the cochain $W$ is clear from the context we write just $S^j$ instead of $S^j(W)$ to denote the fat $j$-faces, as well as $F^j$ and $\Gamma$ instead of $F^j(W)$ and $\Gamma(W)$, respectively.
	\item We relate to each cochain as the event that the matching random variable is in the cochain. For instance, instead of writing $\Pr[P_j \in S^j]$ and $\Pr[P_j \notin S^j]$ we write just $\Pr[S^j]$ and $\Pr[\overline{S^j}]$, respectively.
\end{itemize}

\subsubsection{Proof of Proposition~\ref{good-dimension-pro}}
Our aim is to show that there exists a dimension $j \le i$ in which many fat $j$-faces contain many non-fat $(j-1)$-faces. We first show that a large norm of fat faces of any dimension imply a large norm of $W$.
\begin{lem}[Large norm of fat faces implies large norm of $W$]\label{fat-faces-upper-bound-lemma}
	Let $X$ be a $d$-dimensional simplicial complex, $0 \le i < d$, $0 < \eta < 1$ a fatness constant. Then for any $i$-cochain $W \in C^i(X)$ and any $-1 \le j \le i$
	$$\norm{W} \ge \eta^{2^{i-j}-1}\norm{S^j}.$$
\end{lem}
\begin{proof}
	Fix an $i$-cochain $W \in C^i(X)$ and some $-1 \le j \le i$. By laws of probability, for any $k \le i$
	\begin{equation}\label{fat-faces-upper-bound-eq}
	\begin{split}
	\big\|S^k\big\| &=
	\Pr\!\big[S^k\big] \ge
	\Pr\!\big[S^k \wedge S^{k-1}\big] =
	\Pr\!\big[S^k \mid S^{k-1}\big]\Pr\!\big[S^{k-1}\big] \\[4pt]&\ge
	\eta^{2^{i-k}}\Pr\!\big[S^{k-1}\big] =
	\eta^{2^{i-k}}\big\|S^{k-1}\big\|,
	\end{split}
	\end{equation}
	where the last inequality follows from the definition of fat faces.
	
	Applying \eqref{fat-faces-upper-bound-eq} iteratively for $k = i,i-1,\dotsc,j+1$ yields
	\begin{equation}\nonumber
	\norm{W} \ge
	\eta\,\eta^2 \dotsm \eta^{2^{i-j-1}}\norm{S^j} =
	\eta^{2^{i-j}-1}\norm{S^j}.
	\end{equation}
\end{proof}

Now we show that for any dimension $j \le i$ the norm of $W$ is bounded by the norm of the fat $j$-faces (which represent fat $i$-faces on fat $(i-1)$-faces on fat $(i-2)$-faces and all the way down to fat $j$-faces) plus the norm of fat $k$-faces which contain non-fat $(k-1)$-faces for all $k \in \{j+1,\dotsc,i\}$.
\begin{lem}[Norm of $W$ as fat faces on fat faces plus fat faces on non-fat faces]\label{cochain-upper-bound-lemma}
	Let $X$ be a $d$-dimensional simplicial complex, $0 \le i < d$. Then for any $i$-cochain $W \in C^i(X)$ and any $-1 \le j \le i$
	$$\norm{W} \le \norm{S^j} + \sum_{k=j+1}^i\Pr\!\big[S^k \wedge \overline{S^{k-1}}\,\big].$$
\end{lem}
\begin{proof}
	Fix an $i$-cochain $W \in C^i(X)$ and some $-1 \le j \le i$. By laws of probability, for any $k \le i$
	\begin{equation}\label{cochain-upper-bound-eq}
	\begin{split}
	\big\|S^k\big\| &=
	\Pr\!\big[S^k\big] =
	\Pr\!\big[S^k \wedge S^{k-1}\big] + \Pr\!\big[S^k \wedge \overline{S^{k-1}}\,\big] \\[4pt]&\le
	\Pr\!\big[S^{k-1}\big] + \Pr\!\big[S^k \wedge \overline{S^{k-1}}\,\big] =
	\big\|S^{k-1}\big\| + \Pr\!\big[S^k(W) \wedge \overline{S^{k-1}}\,\big].
	\end{split}
	\end{equation}
	
	Applying \eqref{cochain-upper-bound-eq} iteratively for $k = i,i-1,\dotsc,j+1$ finishes the proof.
\end{proof}

Now we use the above two lemmas in order to prove the proposition.

\begin{proof}[Proof of Proposition~\ref{good-dimension-pro}]
	Fix an $i$-cochain $W \in C^i(X)$, $\norm{W} < \eta^{2^{i+1}-1}$. First we show that the empty set is not fat. By Lemma~\ref{fat-faces-upper-bound-lemma} it holds that
	\begin{equation}\nonumber
	\norm{W} \ge
	\eta^{2^{i+1}-1}\norm{S^{-1}},
	\end{equation}
	
	Since $\norm{W} < \eta^{2^{i+1}-1}$ it follows that $\norm{S^{-1}} < 1$, which implies that $\norm{S^{-1}} = 0$ because there is only one $(-1)$-face (the empty set). Now, if for any $0 \le j \le i$
	\begin{equation}\nonumber
	\Pr\!\big[S^j \wedge \overline{S^{j-1}}\,\big] < \frac{c^j}{i+1}\norm{W},
	\end{equation}
	then by Lemma~\ref{cochain-upper-bound-lemma}
	\begin{equation}\nonumber
	\norm{W} \le
	\sum_{j=0}^i\Pr\!\big[S^j \wedge \overline{S^{j-1}}\,\big] <
	\sum_{j=0}^i \frac{c^j}{i+1}\norm{W} \le
	\norm{W},
	\end{equation}
	which leads to a contradiction.
\end{proof}

\subsubsection{Proof of Proposition~\ref{container-bound-pro}}
Fix an $i$-cochain $W \in C^i(X)$ and some $0 \le j \le i$. By laws of probability it follows that
\begin{equation}\label{container-bound-eq1}
\Pr\!\big[\Gamma \wedge \overline{S^{j-1}}\,\big] \ge
\Pr\!\big[\Gamma \wedge \overline{S^{j-1}} \wedge W\big] =
\Pr\!\big[\Gamma \mid \overline{S^{j-1}} \wedge W\big]\Pr\!\big[\overline{S^{j-1}} \wedge W\big] =
\Pr\!\big[\overline{S^{j-1}} \wedge W\big],
\end{equation}
where the last equality holds since $\Pr[\Gamma \mid W] = 1$.

By an argument similar to Lemma~\ref{fat-faces-upper-bound-lemma} (just add the event $\overline{S^{j-1}}$ to each step) we get that
\begin{equation}\label{container-bound-eq2}
\Pr\!\big[W \wedge \overline{S^{j-1}}\,\big] \ge
\eta^{2^{i-j}-1}\Pr\!\big[S^j \wedge \overline{S^{j-1}}\,\big].
\end{equation}

Combining~\eqref{container-bound-eq1} and~\eqref{container-bound-eq2} finishes the proof.
\qed

\subsubsection{Proof of Proposition~\ref{full-faces-upper-bound-pro}}
We want to show that many of the $(i+1)$-faces, which contain non-fat $(j-1)$-faces, are not full. First we count the full $(i+1)$-faces by two terms: Full $(i+1)$-faces which contain only full faces, and full $(i+1)$-faces which contain a non-full $k$-face for some $k \in \{j,\dotsc,i\}$.
\begin{lem}[Full $(i+1)$-faces as full faces on full faces plus full faces on non-full faces]\label{full-on-non-fat-bound-lemma}
	Let $X$ be a $d$-dimensional simplicial complex, $0 \le i < d$, $0 < \eta < 1$ a fatness constant. Then for any $i$-cochain $W \in C^i(X)$ and any $0 \le j \le i$
	$$\Pr\!\big[F^{i+1} \wedge \overline{S^{j-1}}\,\big] \le
	\Pr\!\big[F^{j+1} \wedge \overline{S^{j-1}}\,\big] + \sum_{k=j+2}^{i+1}\Pr\!\big[F^k \wedge \overline{F^{k-1}}\,\big].$$
\end{lem}
\begin{proof}
	Fix an $i$-cochain $W \in C^i(X)$ and some $0 \le j < i$. By laws of probability, for any $k>j$
	\begin{equation}\label{full-on-non-fat-bound-eq1}
	\begin{split}
	\Pr\!\big[F^k \wedge \overline{S^{j-1}}\,\big] &=
	\Pr\!\big[F^k \wedge \overline{S^{j-1}} \wedge F^{k-1}\big] + \Pr\!\big[F^k \wedge \overline{S^{j-1}} \wedge \overline{F^{k-1}}\,\big] \\&\le
	\Pr\!\big[F^{k-1} \wedge \overline{S^{j-1}}\,\big] + \Pr\!\big[F^k \wedge \overline{F^{k-1}}\,\big].
	\end{split}
	\end{equation}
	
	Applying~\eqref{full-on-non-fat-bound-eq1} iteratively for $k=i+1,i,\dotsc,j+2$ finishes the proof.
\end{proof}

Now we want to bound the second term of Lemma~\ref{full-on-non-fat-bound-lemma}, so we relate the norm of full $k$-faces which contain a non-full $(k-1)$-face to the norm of full $k$-faces which contain a non-fat $(k-2)$-face.
\begin{lem}[Full $k$-faces with a non-full $(k-1)$-face related to full $k$-faces with a non-fat $(k-2)$-face]\label{full-contain-not-full-bound-lemma}
	Let $X$ be a $d$-dimensional simplicial complex, $0 \le i < d$, $0 < \eta < 1$ a fatness constant. Then for any $i$-cochain $W \in C^i(X)$ and any $1 \le k \le i+1$
	$$\Pr\!\big[F^k \wedge \overline{F^{k-1}}\,\big] \le k\Pr\!\big[F^k \wedge \overline{S^{k-2}}\,\big].$$
\end{lem}
\begin{proof}
	Fix an $i$-cochain $W \in C^i(X)$ and some $1 \le k \le i+1$. Note that if a $(k-1)$-face $\sigma \in X(k-1)$ is not full, then by definition it contains a non-fat $(k-2)$-face $\tau \subset \sigma$, $\tau \in X(k-2) \setminus S^{k-2}$. Since each $(k-1)$-face contains $k$ vertices, by removing one vertex from $\sigma$ uniformly at random there is a probability of $1/k$ to hit $\tau$. It follows that
	\begin{equation}\label{full-on-non-fat-bound-eq3}
	\begin{split}
	1 &=
	\Pr\!\big[\overline{F^{k-1}} \mid F^k \wedge \overline{S^{k-2}}\,\big] =
	\frac{\Pr\!\big[\overline{F^{k-1}} \wedge F^k \wedge \overline{S^{k-2}}\,\big]}{\Pr\!\big[F^k \wedge \overline{S^{k-2}}\,\big]} \\[2pt]&=
	\frac{\Pr\!\big[\overline{S^{k-2}} \mid F^k \wedge \overline{F^{k-1}}\,\big]\Pr\!\big[F^k \wedge \overline{F^{k-1}}\,\big]}{\Pr\!\big[F^k \wedge \overline{S^{k-2}}\,\big]} \ge
	\frac{\Pr\!\big[F^k \wedge \overline{F^{k-1}}\,\big]}{k\Pr\!\big[F^k \wedge \overline{S^{k-2}}\,\big]},
	\end{split}
	\end{equation}
	where the first equality holds since if $P_{k-2} \in \overline{S^{k-2}}$, then by definition $P_{k-1}$ must have been not full, the second and third equalities follow from laws of probability and the inequality follows from the explanation above.
\end{proof}

The next lemma is the {\em only} place where we use the skeleton expansion of the complex. We use it in order to bound the norm of full $k$-faces which contain a non-fat $(k-2)$-face with relation to the norm of fat $(k-1)$-faces which contain a non-fat $(k-2)$-face.
\begin{lem}[Full $k$-faces with a non-fat $(k-2)$-face related to fat $(k-1)$-faces with a non-fat $(k-2)$-face]\label{k-full-and-k-2-not-fat-lemma}
	Let $X$ be a $d$-dimensional $\alpha$-skeleton expander, $0 \le i < d$, $0 < \eta < 1$ a fatness constant. Then for any $i$-cochain $W \in C^i(X)$ and any $1 \le k \le i+1$
	$$\Pr\!\big[F^k \wedge \overline{S^{k-2}}\,\big] \le \big(\eta^{2^{i-k+1}}+\alpha\big)\!\Pr\!\big[S^{k-1} \wedge \overline{S^{k-2}}\,\big].$$
\end{lem}
\begin{proof}
	Fix an $i$-cochain $W \in C^i(X)$ and some $1 \le k \le i+1$. Recall that a $k$-face is full if all of its $(k-1)$-faces are fat. Thus, at the link of any $(k-2)$-face, a full $k$-face seems as an edge between two fat vertices. It follows that for any non-fat $(k-2)$-face $\sigma \in X(k-2) \setminus S^{k-2}$
	\begin{equation}\nonumber
	\big\|F^k_\sigma\big\|_\sigma \le
	\big\|E(S^{k-1}_\sigma)\big\|_\sigma \le
	\big(\eta^{2^{i-k+1}} + \alpha\big)\big\|S^{k-1}_\sigma\big\|_\sigma,
	\end{equation}
	where the second inequality follows from the skeleton expansion of $X$ and that $\sigma$ is not fat. Then by the law of total probability
	\begin{equation}\nonumber
	\begin{split}
	\Pr\!\big[F^k \wedge \overline{S^{k-2}}\,\big] &=
	\sum_{\sigma \in X(k-2) \setminus S^{k-2}}\Pr\!\big[F^k \mid \sigma\big]\Pr\!\big[\sigma\big] \\&=
	\sum_{\sigma \in X(k-2) \setminus S^{k-2}}\big\|F^k_\sigma\big\|_\sigma\Pr\!\big[\sigma\big] \\&\le
	\sum_{\sigma \in X(k-2) \setminus S^{k-2}}\big(\eta^{2^{i-k+1}} + \alpha\big)\big\|S^{k-1}_\sigma\big\|_\sigma\Pr\!\big[\sigma\big] \\&=
	\big(\eta^{2^{i-k+1}} + \alpha\big)\sum_{\sigma \in X(k-2) \setminus S^{k-2}}\Pr\!\big[S^{k-1} \mid \sigma\big]\Pr\!\big[\sigma\big] \\&=
	\big(\eta^{2^{i-k+1}} + \alpha\big)\Pr\!\big[S^{k-1} \wedge \overline{S^{k-2}}\big].
	\end{split}
	\end{equation}
\end{proof}

Now the proposition follows as an immediate corollary of the above lemmas.
\begin{proof}[Proof of Proposition~\ref{full-faces-upper-bound-pro}]
	Fix an $i$-cochain $W \in C^i(X)$ and some $0 \le j \le i$. Then by the above lemmas
	\begin{displaymath}
	\begin{split}
	\Pr\!\big[F^{i+1} \wedge \overline{S^{j-1}}\,\big] &\le
	\Pr\!\big[F^{j+1} \wedge \overline{S^{j-1}}\,\big] + \sum_{k=j+2}^{i+1}\Pr\!\big[F^k \wedge \overline{F^{k-1}}\,\big] \\&\le
	\Pr\!\big[F^{j+1} \wedge \overline{S^{j-1}}\,\big] + \sum_{k=j+2}^{i+1}k\Pr\!\big[F^k \wedge \overline{S^{k-2}}\,\big] \\&\le
	\big(\eta^{2^{i-j}} + \alpha\big)\Pr\!\big[S^j \wedge \overline{S^{j-1}}\,\big] + \sum_{k=j+1}^i(k+1)\big(\eta^{2^{i-k}} + \alpha\big)\Pr\!\big[S^k \wedge \overline{S^{k-1}}\,\big],
	\end{split}
	\end{displaymath}
	where the inequalities follow from Lemma~\ref{full-on-non-fat-bound-lemma}, Lemma~\ref{full-contain-not-full-bound-lemma} and Lemma~\ref{k-full-and-k-2-not-fat-lemma}, in that order.
\end{proof}

\section{Rapid mixing of high order random walks}
In this section we prove the following two theorems.
\begin{thm}[Colorful expansion implies rapid mixing]\label{main-thm-2}
	Let $X$ be a $d$-dimensional $\epsilon$-colorful expander, $d>1$. Then all of the high order random walks on $X$ are $\mu$-rapidly mixing for
	$$\mu = 1 - \frac{\epsilon^2}{2(d+1)^2}.$$
\end{thm}

\begin{thm}[Rapid mixing criterion]\label{main-thm-3}
	Let $X$ be a $d$-dimensional $\alpha$-skeleton expander, $d>1$. If $\alpha < (\sqrt[2^d]2 - 1)\sqrt2$, then all of the high order random walks on $X$ are $\mu$-rapidly mixing for
	$$\mu = 1 - \frac{1}{2(d+1)^2}\left(\frac{\sqrt[2^d]2-1-\sqrt2\alpha}{2\sqrt2d}\right)^{2d}.$$
\end{thm}

We note that by following the steps of the proof carefully we actually get a stronger theorem, which we state here without proving it.
\begin{thm}[Stronger rapid mixing]
	Let $X$ be a $d$-dimensional $\alpha$-skeleton expander, $d>1$. Then for any $0 \le i < d$ such that $\alpha < (\sqrt[2^{i+1}]2-1)/\sqrt2$ the high order random walk on dimension $i$ of $X$ is $\mu$-rapidly mixing for
	$$\mu = 1 - \frac{1}{2(i+2)^2}\left(\frac{\sqrt[2^{i+1}]2-1-\sqrt2\alpha}{2\sqrt2(i+1)}\right)^{2(i+1)}.$$
\end{thm}

Recall that the mixing rate of a random walk is deduced from the spectrum of its normalized adjacency matrix.
We show in the following lemmas that all of the induced $i$-graphs, $0 \le i < d$, have a large conductance and a large bipartiteness ratio, which imply the concentration of the spectrum of the normalized adjacency matrices $\widetilde A_i$ for all $0 \le i < d$.

\begin{lem}\label{bound-conductance-lemma}
	Let $X$ be a $d$-dimensional $\epsilon$-colorful expander. Then for any $0 \le i < d$ the conductance of $G_i(X)$ satisfies
	$$\Phi(G_i(X)) \ge \frac{\epsilon}{i+2}.$$
\end{lem}
\begin{proof}
	Fix $0 \le i < d$ and denote by $G_i = (V_i,E_i)$ the induced $i$-graph of $X$. For any $i$-face $\sigma \in X(i)$ and its corresponding vertex $v \in V_i$ it holds that
	\begin{equation}\label{deg-vol-eq}
	\deg(\sigma) =
	\frac{1}{d-i}\sum_{\substack{\tau \supset \sigma \\ |\tau| = |\sigma|+1}}\deg(\tau) =
	\frac{1}{(d-i)(i+1)}\sum_{\sigma' \sim \sigma}\deg(\sigma \cup \sigma') =
	\frac{\deg(v)}{(d-i)(i+1)},
	\end{equation}
	where the first equality holds since each $d$-face which contains $\sigma$ contains $d-i$ $(i+1)$-faces $\tau \supset \sigma$, and hence is counted $d-i$ times. And the second equality holds since each $(i+1)$-face which contains $\sigma$ contains $i+1$ neighbors of $\sigma$.
	
	Now let $\emptyset \ne S \subseteq V_i$, $\vol(S)/\vol(V_i) \le 1/2$, be an arbitrary subset of vertices in $G_i$ and denote by $W \in C^i(X)$ the corresponding $i$-cochain in $X$. By~\eqref{deg-vol-eq} it follows that
	\begin{equation}\nonumber
	\norm{W} =
	\frac{\sum_{\sigma \in W}\deg(\sigma)}{\sum_{\sigma \in X(i)}\deg(\sigma)} =
	\frac{\sum_{v \in S}\deg(v)}{\sum_{v \in V_i}\deg(v)} =
	\frac{\vol(S)}{\vol(V_i)} \le
	\frac{1}{2},
	\end{equation}
	and thus by the $\epsilon$-colorful expansion of $X$ we get that
	\begin{equation}\label{colorful-exp-to-conductance-eq1}
	\frac{\norm{\psi(W)}}{\norm{W}} \ge
	\epsilon.
	\end{equation}
	
	We claim that any expanding face $\tau \in \psi(W)$ contributes at least $(i+1)\deg(\tau)$ edges between $S$ and $\bar S$. In order to see this, consider an expanding face $\tau \in \psi(W)$ and let $j = |\{\sigma \subset \tau \;|\; \sigma \in W\}|$ denote the number of $i$-faces of $\tau$ which are in $W$. Since the total number of $i$-faces in $\tau$ is $i+2$, then $\tau$ has $(i+2-j)$ $i$-faces which are not in $W$. It follows that there are $j(i+2-j)$ pairs of $i$-faces $\sigma,\sigma' \subset \tau$ such that $\sigma \in W$, $\sigma' \notin W$. Recall that for each such pair, there are $\deg(\tau)$ edges between the corresponding vertices in $G_i$, so $\tau$ contributes $j(i+2-j)\deg(\tau)$ edges between $S$ and $\bar S$. The last thing to note is that $\tau$ is an expanding face so $1 \le j \le i+1$, which yields that $j(i+2-j) \ge i+1$. Therefore,
	\begin{equation}\label{colorful-exp-to-conductance-eq2}
	|E(S,\bar S)| \ge
	(i+1)\sum_{\tau \in \psi(W)}\deg(\tau).
	\end{equation}
	
	It follows that
	\begin{equation}\label{colorful-exp-to-conductance-eq3}
	\begin{split}
	\frac{\norm{\psi(W)}}{\norm{W}} &=
	\frac{\sum_{\tau \in \psi(W)}\deg(\tau)}{\sum_{\tau \in X(i+1)}\deg(\tau)}\cdot\frac{\sum_{\sigma \in X(i)}\deg(\sigma)}{\sum_{\sigma \in W}\deg(\sigma)} = \frac{i+2}{d-i}\cdot\frac{\sum_{\tau \in \psi(W)}\deg(\tau)}{\sum_{\sigma \in W}\deg(\sigma)} \\&\le
	(i+2)\frac{|E(S,\bar S)|}{\vol(S)} =
	(i+2)\Phi(S),
	\end{split}
	\end{equation}
	where the second equality holds since
	\begin{equation}\nonumber
	\frac{\sum_{\sigma \in X(i)}\deg(\sigma)}{\sum_{\tau \in X(i+1)}\deg(\tau)} =
	\frac{\displaystyle{\binom{d+1}{i+1}|X(d)|}}{\displaystyle{\binom{d+1}{i+2}|X(d)|}} =
	\frac{i+2}{d-i},
	\end{equation}
	and the inequality holds by~\eqref{deg-vol-eq} and~\eqref{colorful-exp-to-conductance-eq2}.
	
	Combining~\eqref{colorful-exp-to-conductance-eq1} and~\eqref{colorful-exp-to-conductance-eq3} finishes the proof.
\end{proof}

\begin{lem}\label{bound-bipartiteness-lemma}
	Let $X$ be a $d$-dimensional simplicial complex. Then for any $1 \le i < d$ the bipartiteness ratio of $G_i(X)$ satisfies
	$$\beta(G_i(X)) \ge \frac{1}{i+2}.$$
\end{lem}
\begin{proof}
	Fix $1 \le i < d$ and denote by $G_i = (V_i,E_i)$ the induced $i$-graph of $X$. Let $S,S_1,S_2 \subseteq V_i$ be arbitrary subsets of vertices such that $S_1 \cupdot S_2 = S$ and denote by $W,W_1,W_2 \in C^i(X)$ the corresponding $i$-cochains in $X$.
	
	For any $(i+1)$-face $\tau \in \Gamma(W)$ denote by $j_1(\tau) = |\{\sigma \subset \tau \;|\; \sigma \in W_1\}|$ the number of $i$-faces of $\tau$ which are in $W_1$. Similarly denote by $j_2(\tau)$ the number of $i$-faces of $\tau$ which are in $W_2$. Since the total number $i$-faces in $\tau$ is $i+2$, then $\tau$ contains $(i+2-j_1(\tau)-j_2(\tau))$ $i$-faces which are not in $W$. It follows that
	\begin{equation}\label{bound-bipartiteness-eq1}
	|E(S,\bar S)| =
	\sum_{\tau \in \Gamma(W)}(j_1(\tau)+j_2(\tau))(i+2-j_1(\tau)-j_2(\tau))\deg(\tau),
	\end{equation}
	and for any $k=1,2$
	\begin{equation}\label{bound-bipartiteness-eq2}
	2|E(S_k)| =
	\sum_{\tau \in \Gamma(W)}2\binom{j_k(\tau)}{2}\deg(\tau).
	\end{equation}
	
	Combining~\eqref{bound-bipartiteness-eq1} and~\eqref{bound-bipartiteness-eq2} yields
	\begin{equation}\label{bound-bipartiteness-eq3}
	|E(S,\bar S)|+2(|E(S_1)|+|E(S_2)|) =
	\sum_{\tau \in \Gamma(W)}\Big(\big(j_1(\tau)+j_2(\tau)\big)(i+1)-2j_1(\tau)j_2(\tau)\Big)\deg(\tau).
	\end{equation}
	
	Consider an arbitrary $(i+1)$-face $\tau \in \Gamma(W)$. If $j_1(\tau)=0$ or $j_2(\tau)=0$, then
	\begin{equation}\label{bound-bipartiteness-eq4}
	\big(j_1(\tau)+j_2(\tau)\big)(i+1)-2j_1(\tau)j_2(\tau) \ge
	i+1.
	\end{equation}
	Otherwise,
	\begin{equation}\label{bound-bipartiteness-eq5}
	\begin{split}
	\big(j_1(\tau)+j_2(\tau)\big)(i+1)-2j_1(\tau)j_2(\tau) \ge
	\big(j_1(\tau)+j_2(\tau)\big)\left(i+1-\frac{j_1(\tau)+j_2(\tau)}{2}\right) \ge
	i+1,
	\end{split}
	\end{equation}
	where the first inequality holds since $j_1(\tau)j_2(\tau) \le ((j_1(\tau)+j_2(\tau))/2)^2$, and the second inequality holds since $2 \le j_1(\tau)+j_2(\tau) \le i+2$.
	
	Substituting~\eqref{bound-bipartiteness-eq4} and~\eqref{bound-bipartiteness-eq5} in~\eqref{bound-bipartiteness-eq3} yields
	\begin{equation}\label{bound-bipartiteness-eq6}
	|E(S,\bar S)|+2(|E(S_1)|+|E(S_2)|) \ge
	(i+1)\sum_{\tau \in \Gamma(W)}\deg(\tau).
	\end{equation}
	
	It also holds that
	\begin{equation}\label{bound-bipartiteness-eq7}
	\begin{split}
	\vol(S) &=
	\sum_{v \in S}\deg(v) =
	\sum_{\sigma \in W}\sum_{\sigma' \sim \sigma}\deg(\sigma \cup \sigma') \\&=
	\sum_{\tau \in \Gamma(W)}\big(j_1(\tau)+j_2(\tau)\big)(i+1)\deg(\tau) \le
	(i+2)(i+1)\sum_{\tau \in \Gamma(W)}\deg(\tau).
	\end{split}
	\end{equation}
	
	Combining~\eqref{bound-bipartiteness-eq6} and~\eqref{bound-bipartiteness-eq7} yields
	\begin{equation}\nonumber
	\beta(S,S_1,S_2) =
	\frac{|E(S,\bar S)|+2(|E(S_1)|+|E(S_2)|)}{\vol(S)} \ge
	\frac{1}{i+2}.
	\end{equation}
	
	Since $S,S_1,S_2$ were arbitrary it follows that
	\begin{equation}\nonumber
	\beta(G_i(X)) \ge
	\frac{1}{i+2},
	\end{equation}
	which finishes the proof.
\end{proof}

\noindent\textbf{Proof of Theorem~\ref{main-thm-2}.}
The proof follows from the previous lemmas.
For any $0 \le i < d$ the combination of Lemma~\ref{cheeger-lambda-2} and Lemma~\ref{bound-conductance-lemma} yield
\begin{equation}\label{main-thm-2-eq1}
\lambda_2(\widetilde A_i) \le
1 - \frac{\Phi(G_i)^2}{2} \le
1 - \frac{\epsilon^2}{2(i+2)^2} \le
1 - \frac{\epsilon^2}{2(d+1)^2}.
\end{equation}

For $1 \le i < d$ the combination of Lemma~\ref{cheeger-lambda-n} and Lemma~\ref{bound-bipartiteness-lemma} yield
\begin{equation}\label{main-thm-2-eq2}
\lambda_{|V_i|}(\widetilde A_i) \ge
-1 + \frac{\beta(G_i)^2}{2} \ge
-1 + \frac{1}{2(i+2)^2} \ge
-1 + \frac{1}{2(d+1)^2} \ge
-1 + \frac{\epsilon^2}{2(d+1)^2},
\end{equation}
where the last inequality holds since $\epsilon<1$.

For $i=0$ denote by $G'_0$ the graph achieved by adding $\deg(v)$ self-loops to any vertex $v \in V_0$. Then by Proposition~\ref{self-loops-prop}
\begin{equation}\label{main-thm-2-eq3}
\lambda_2(\widetilde A'_0) =
\frac{1}{2}(1+\lambda_2(\widetilde A_0)) \le
\frac{1}{2}\left(1+1-\frac{\Phi(G_0)^2}{2}\right) \le
1 - \frac{\epsilon^2}{4\cdot2^2} \le
1-\frac{\epsilon^2}{2(d+1)^2},
\end{equation}
where the first and second inequalities follow from Lemma~\ref{cheeger-lambda-2} and Lemma~\ref{bound-conductance-lemma} respectively, and the last inequality holds since $d>1$. By the same proposition it also holds that $\lambda_{|V_0|}(\widetilde A'_0) \ge 0$.

By~\eqref{main-thm-2-eq1},~\eqref{main-thm-2-eq2} and~\eqref{main-thm-2-eq3} it follows that for any $0 \le i < d$ the high order random walk on dimension $i$ of $X$ is $\mu$-rapidly mixing for
$$\mu = 1 - \frac{\epsilon^2}{2(d+1)^2}.$$
\qed
\bigskip

\noindent\textbf{Proof of Theorem~\ref{main-thm-3}.}
The proof follows immediately as a corollary of Theorem~\ref{main-thm-1} and Theorem~\ref{main-thm-2}.
\qed

\section{Explicit construction}
Ramanujan complexes were first defined in~\cite{LSV05.1} and were explicitly constructed in~\cite{LSV05.2}. For details on Ramanujan complexes we refer the readers to~\cite{Lub14}.

In this section we prove the following theorem.
\begin{thm}[Explicit construction]\label{main-thm-4}
	For any $d \in \N$ there exists a constant $q_0 = q_0(d)$ such that if $X$ is a $d$-dimensional $q$-thick Ramanujan complex with $q > q_0$, then there exists a constant $\mu = \mu(d,q)$ such that all of the high order random walks on $X$ are $\mu$-rapidly mixing.
\end{thm}

As been proven in~\cite{EK16}, Ramanujan complexes are excellent skeleton expanders. Though, we need a stronger claim for the mixing of their $1$-dimensional skeletons than the one appears in~\cite{EK16}. We can get a better bound since we need a good mixing behavior only inside a subset of vertices and not between every two subsets. We start by defining a special type of complexes.

\begin{dfn}[Partite regular complex]
	Let $X$ be a $d$-dimensional simplicial complex. $X$ is said to be {\em partite regular} if there exists a partition of its vertices to disjoint subsets $V_0 \cupdot V_1 \cupdot \dotsb \cupdot V_d = X(0)$ such that:
	\begin{itemize}
		\item For any $d$-dimensional face $\sigma \in X(d)$ it holds that $\sigma \in V_0 \times V_1 \times \dotsb \times V_d$.
		\item For any $I \subset J \subset \{0,1,\dotsc,d\}$ there exists $k_I^J \in \mathbb N$ such that any face $\sigma \in X \cap \prod_{i \in I}V_i$ is contained in $k_I^J$ faces from $X \cap \prod_{j \in J}V_j$.
	\end{itemize}
\end{dfn}

For a partite regular complex $X$, denote by $X_{(i,j)} = (V_i \cupdot V_j, X(1) \cap (V_i \times V_j))$ the induced graph by partitions $i$ and $j$. Note that by the regularity of the complex, $X_{(i,j)}$ is a bipartite biregular graph, i.e., there exists $k_i^j, k_j^i \in \N$, such that any vertex $v \in V_i$ is contained in $k_i^j$ edges and any vertex $u \in V_j$ is contained in $k_j^i$ edges. It is known that $\lambda_1(X_{(i,j)}) = (k_i^jk_j^i)^{\rfrac{1}{2}}$, where $\lambda_1(X_{(i,j)})$ is the largest eigenvalue of the graph's adjacency matrix (see~\cite[Lemma~3.1]{EGL15} for a proof). Denote by $\widetilde \lambda_2(X_{(i,j)}) = \lambda_2(X_{(i,j)})/(k_i^jk_j^i)^{\rfrac{1}{2}}$ the normalized second largest eigenvalue. The following mixing lemma for bipartite biregular graphs is proven in~\cite{EGL15}.

\begin{lem}\label{bipartite-mixing-lemma}\cite[Corollary~3.4]{EGL15}
	Let $G=(V_1 \cupdot V_2, E)$ be a bipartite biregular graph, $\widetilde \lambda_2(G)$ its normalized second largest eigenvalue. Then for any $S \subseteq V_1, T \subseteq V_2$
	$$\frac{|E(S,T)|}{|E|} \le \sqrt{\frac{|S||T|}{|V_1||V_2|}}\left(\sqrt{\frac{|S||T|}{|V_1||V_2|}} + \widetilde \lambda_2(G)\right).$$
\end{lem}

We use this lemma in order to prove the following proposition.
\begin{pro}[Skeleton mixing lemma]\label{skeleton-mixing-lemma}
	Let $X$ be a partite regular $d$-dimensional complex, $\widetilde \lambda_2(X) = \max_{0 \le i < j \le d}\widetilde \lambda_2(X_{(i,j)})$. Then for any subset of vertices $S \subseteq X(0)$
	$$\norm{E(S)} \le \norm{S}(\norm{S} + \widetilde \lambda_2(X)).$$
\end{pro}
\begin{proof}
	Let $V_0 \cupdot V_1 \cupdot \dotsb \cupdot V_d$ be the partition of $X(0)$. Note that since $X$ is a partite regular complex, then for any $0 \le i \le d$ and any two vertices $u,v \in V_i$ it holds that $\deg(u) = \deg(v)$. It is achieved by the regularity property of the complex and by setting $I = \{i\}$ and $J=[d]=\{0,1,\dotsc,d\}$. This implies that for any $0 \le i \le d$
	\begin{equation}\nonumber
	\sum_{v \in X(0)}\deg(v) =
	(d+1)|X(d)| =
	(d+1)\sum_{v \in V_i}deg(v) =
	(d+1)|V_i|k_{\{i\}}^{[d]}.
	\end{equation}
	
	Therefore, for any subset of vertices $S_i \subseteq V_i$
	\begin{equation}\label{norm-of-vertices}
	\norm{S_i} =
	\frac{\sum_{v \in S_i}deg(v)}{\sum_{u \in X(0)}deg(u)} =
	\frac{|S_i|}{(d+1)|V_i|}
	\end{equation}
	
	In a same way, for any $0 \le i < j \le d$
	\begin{equation}\nonumber
	\begin{split}
	\sum_{e \in X(1)}\deg(e) &=
	\binom{d+1}{2}|X(d)| =
	\binom{d+1}{2}\sum_{e \in X(1) \cap (V_i \times V_j)}deg(e) \\&=
	\binom{d+1}{2}|X(1) \cap (V_i \times V_j)|k_{\{i,j\}}^{[d]}.
	\end{split}
	\end{equation}
	
	And again for any subset of edges $S_{ij} \subseteq X(1) \cap (V_i \times V_j)$
	\begin{equation}\label{norm-of-edges}
	\norm{S_{ij}} =
	\sum_{e \in S_{ij}}w(e) =
	\frac{\sum_{e \in S_{ij}}deg(e)}{\sum_{f \in X(1)}deg(u)} =
	\frac{2|S_{ij}|}{d(d+1)|X(1) \cap (V_i \times V_j)|}.
	\end{equation}
	
	Now let $S \subseteq X(0)$ be an arbitrary subset of vertices in the complex. For any $0 \le i \le d$ denote by $S_i = S \cap V_i$ the vertices in partition $i$. For any $0 \le i < j \le d$, \eqref{norm-of-vertices} yields
	\begin{equation}\nonumber
	\sqrt{\frac{|S_i||S_j|}{|V_i||V_j|}} =
	\sqrt{(d+1)^2\norm{S_i}\norm{S_j}} =
	(d+1)\sqrt{\norm{S_i}\norm{S_j}},
	\end{equation}
	and \eqref{norm-of-edges} yields
	\begin{equation}\nonumber
	\frac{|E(S_i,S_j)|}{|X(1) \cap (V_i \times V_j)|} =
	\frac{d(d+1)}{2}\norm{E(S_i,S_j)}.
	\end{equation}
	
	By Lemma~\ref{bipartite-mixing-lemma} it follows that for any $0 \le i < j \le d$
	\begin{equation}\nonumber
	\norm{E(S_i,S_j)} \le
	\frac{2}{d}\sqrt{\norm{S_i}\norm{S_j}}\left((d+1)\sqrt{\norm{S_i}\norm{S_j}} + \widetilde \lambda_2(X)\right).
	\end{equation}
	
	Note that since $S = S_0 \cupdot S_1 \cupdot \dotsb \cupdot S_d$, then $\norm{S} = \sum_{i=0}^d\norm{S_i}$. Thus, the sum $\sum_{i \ne j}\norm{S_i}\norm{S_j}$ is maximized when all of the subsets are equal, i.e., $\norm{S_i} = \norm{S}/(d+1)$ for all $0 \le i \le d$. It follows that
	\begin{equation}
	\begin{split}
	\norm{E(S)} &=
	\sum_{i \ne j}\norm{E(S_i, S_j)} \le
	\sum_{i \ne j}\frac{2}{d}\sqrt{\norm{S_i}\norm{S_j}}\left((d+1)\sqrt{\norm{S_i}\norm{S_j}} + \widetilde \lambda_2(X)\right) \\[5pt] &\le
	\sum_{i \ne j}\frac{2}{d}\sqrt{\frac{\norm{S}^2}{(d+1)^2}}\left((d+1)\sqrt{\frac{\norm{S}^2}{(d+1)^2}} + \widetilde \lambda_2(X)\right) \\[5pt] &=
	\binom{d+1}{2}\frac{2}{d(d+1)}(\norm{S} + \widetilde \lambda_2(X)) =
	\norm{S}(\norm{S} + \widetilde \lambda_2(X)),
	\end{split}
	\end{equation}
	which finishes the proof.
\end{proof}

Now, as been proven in~\cite{EK16}, all of the links of a Ramanujan complex are partite regular and the normalized second largest eigenvalue of every induced bipartite biregular graph approaches $0$ as a function of the dimension and the thickness of the complex. So we state the following lemma which is proven in~\cite{EK16}.
\begin{lem}\label{ramanujan-lambda-bound}
	Let $X$ be a $d$-dimensional $q$-thick Ramanujan complex.
	Then there exists a constant $C=C(d)$ such that
	$$\max_{\substack{\sigma \in X \\ i \ne j}}\widetilde \lambda_2((X_\sigma)_{(i,j)}) \le \frac{C}{\sqrt{q}}.$$
\end{lem}

We are now ready to prove the theorem of this section.
\bigskip

\noindent\textbf{Proof of Theorem~\ref{main-thm-4}.}
Let $d \in \N$ be any dimension we want.
Let $C=C(d)$ be the constant from Lemma~\ref{ramanujan-lambda-bound} and set
\begin{equation}\nonumber
q_0 =
\left(\frac{\sqrt2C}{\sqrt[2^d]2-1}\right)^2.
\end{equation}

Now if $X$ is a $d$-dimensional $q$-thick Ramanujan complex with $q > q_0$, then by Proposition~\ref{skeleton-mixing-lemma} and Lemma~\ref{ramanujan-lambda-bound} $X$ is an $\alpha$-skeleton expander for
\begin{equation}
\alpha =
\frac{C}{\sqrt q} <
\frac{\sqrt[2^d]2-1}{\sqrt2}.
\end{equation}
Then applying Theorem~\ref{main-thm-3} finishes the proof.
\qed

\appendix
\section{Mixing rate for general graphs}
We show here that the mixing rate of random walks on a general graph can be deduced from the spectrum of its normalized adjacency matrix. The idea is similar to the proof for regular graphs, but the details require some more carefulness.

\begin{pro}\label{random-walk-mixing-thm}
	Let $G=(V,E)$ be an undirected graph on $n$ vertices, $\widetilde A$ its normalized adjacency matrix, $\lambda_1 \ge \lambda_2 \ge \dotsb \ge \lambda_n$ the spectrum of $\widetilde A$ and $\lambda = \max\{|\lambda_2|,|\lambda_n|\}$. Then for any initial probability distribution $\pi_0 \in \R^V$ and any $t \in \N$
	$$\norm{\pi_t - \pi}_2 \le \sqrt\frac{d_{max}}{d_{min}}\lambda^t,$$
	where $\pi_t$ is the probability distribution after $t$ steps of the random walk, $\pi$ is the stationary distribution, $d_{max}=\max_{v \in V}\{\deg(v)\}$ and $d_{min}=\min_{v \in V}\{\deg(v)$\}.
\end{pro}
\begin{proof}
	Recall that $\widetilde A=D^{\rfrac{1}{2}}AD^{\rfrac{1}{2}}$, where $D=\mbox{diag}(1/\deg(v))_{v \in V}$ and $A$ is the adjacency matrix of $G$. It follows that for any initial probability distribution $\pi_0 \in \R^V$ and any $t \in \N$
	$$\pi_t = \pi_0(DA)^t=\pi_0(D^{\rfrac{1}{2}}\widetilde AD^{-\rfrac{1}{2}})^t=\pi_0D^{\rfrac{1}{2}}\widetilde A^tD^{-\rfrac{1}{2}}.$$
	
	We show first that the claim holds for a random walk starting from any fixed vertex and then we extend it to any initial probability distribution on the vertices.
	
	Fix a starting vertex $v \in V$ and denote by $\one_v$ the probability distribution which is fixed on $v$, i.e., $\one_v(v) = 1$ and for any other vertex $u \ne v$, $\one_v(u) = 0$. Since $\widetilde A$ is a symmetric matrix, its eigenvectors $\{\mathbf e_1, \dotsc, \mathbf e_n\}$ form an orthonormal basis for $\R^n$, and $\widetilde A$ has the spectral representation
	$$\widetilde A = \sum_{i=1}^n\lambda_i\mathbf e_i^T \mathbf e_i = \sum_{i=1}^n\lambda_iE_i,$$
	where $E_i = \mathbf e_i^T \mathbf e_i$. Note that for any $i$, $E_i^2 = E_i$ and for any $i \ne j$, $E_iE_j = 0$. It follows that
	\begin{equation}\label{spectral-rep-eq}
	\one_v(DA)^t =
	\one_vD^{\rfrac{1}{2}}\widetilde A^tD^{-\rfrac{1}{2}} =
	\one_vD^{\rfrac{1}{2}}\Big(\sum_{i=1}^n\lambda_i^tE_i\Big)D^{-\rfrac{1}{2}} =
	\sum_{i=1}^n\lambda_i^t\one_vD^{\rfrac{1}{2}}E_iD^{-\rfrac{1}{2}}.
	\end{equation}
	
	The first eigenvalue of $\widetilde A$ is the trivial one, i.e., $\lambda_1 = 1$, with corresponding eigenvector
	$$\mathbf e_1 = \left(\sqrt{\frac{\deg(u)}{\sum_{w \in V}\deg(w)}}\right)_{u \in V}.$$
	Note that for any $1 \le i \le n$, $\one_vE_i = \mathbf e_i(v)\mathbf e_i$. Then by substitution we get
	\begin{equation}\label{spectral-rep-first-component-eq}
	\lambda_1^t\one_vD^{\rfrac{1}{2}}E_1D^{-\rfrac{1}{2}} =
	\frac{\one_vE_1D^{-\rfrac{1}{2}}}{\sqrt{\deg(v)}} =
	\frac{\mathbf e_1(v)\mathbf e_1D^{-\rfrac{1}{2}}}{\sqrt{\deg(v)}} =
	\pi.
	\end{equation}
	
	Combining \eqref{spectral-rep-eq} and \eqref{spectral-rep-first-component-eq} yields
	\begin{equation}\nonumber
	\begin{split}
	\norm{\one_v(DA)^t - \pi}_2 &=
	\norm{\sum_{i=2}^n\lambda_i^t\one_vD^{\rfrac{1}{2}}E_iD^{-\rfrac{1}{2}}}_2 \le
	\sqrt{\frac{d_{max}}{d_{min}}}\norm{\sum_{i=2}^n\lambda_i^t\one_vE_i}_2 \\[5pt]&=
	\sqrt{\frac{d_{max}}{d_{min}}}\norm{\sum_{i=2}^n\lambda_i^t\mathbf e_i(v)\mathbf e_i}_2 =
	\sqrt{\frac{d_{max}}{d_{min}}}\sqrt{\inproduct{\sum_{i=2}^n\lambda_i^t\mathbf e_i(v)\mathbf e_i}{\sum_{j=2}^n\lambda_j^t\mathbf e_j(v)\mathbf e_j}} \\[5pt]&=
	\sqrt{\frac{d_{max}}{d_{min}}}\sqrt{\sum_{i=2}^n\lambda_i^{2t}\mathbf e_i(v)^2} \le
	\sqrt{\frac{d_{max}}{d_{min}}}\lambda^t\sqrt{\sum_{i=2}^n\mathbf e_i(v)^2} \le
	\sqrt{\frac{d_{max}}{d_{min}}}\lambda^t,
	\end{split}
	\end{equation}
	where the last equality and the last inequality follow from the orthonormality of the eigenvectors.
	
	Now let $\pi_0 \in \R^V$ be any initial probability distribution on the vertices. We can write $\pi_0$ as a convex combination $\pi_0 = \sum_{v \in V}\alpha_v\one_v$, where $\sum_{v \in V}\alpha_v = 1$. Similarly, the stationary distribution can be written as $\pi = \sum_{v \in V}\alpha_v\pi$.
	Then by the triangle inequality we get
	\begin{equation}\nonumber
	\begin{split}
	\norm{\pi_t - \pi}_2 &=
	\norm{\pi_0(DA)^t - \pi}_2 =
	\norm{\sum_{v \in V}\alpha_v\left(\one_v(DA)^t - \pi\right)}_2 \\&\le
	\sum_{v \in V}\alpha_v\norm{\one_v(DA)^t - \pi}_2 \le
	\sqrt{\frac{d_{max}}{d_{min}}}\lambda^t,
	\end{split}
	\end{equation}
	which finishes the proof.
\end{proof}


\begin{thebibliography}{10}
	
	\bibitem{Alo86}
	N.~Alon.
	\newblock Eigenvalues and expanders.
	\newblock {\em Combinatorica}, 6(2):83--96, 1986.
	
	\bibitem{Din07}
	I.~Dinur.
	\newblock The pcp theorem by gap amplification.
	\newblock {\em Journal of the ACM (JACM)}, 54(3):12, 2007.
	
	\bibitem{EGL15}
	S.~Evra, K.~Golubev, and A.~Lubotzky.
	\newblock {Mixing Properties and the Chromatic Number of Ramanujan Complexes}.
	\newblock {\em International Mathematics Research Notices},
	2015(22):11520--11548, 2015.
	
	\bibitem{EK16}
	S.~Evra and T.~Kaufman.
	\newblock Bounded degree cosystolic expanders of every dimension.
	\newblock In {\em Proceedings of the 48th Annual {ACM} {SIGACT} Symposium on
		Theory of Computing, {STOC} 2016, Cambridge, MA, USA, June 18-21, 2016},
	pages 36--48, 2016.
	
	\bibitem{Gar73}
	H.~Garland.
	\newblock {p-Adic Curvature and the Cohomology of Discrete Subgroups of p-Adic
		Groups}.
	\newblock {\em Annals of Mathematics}, 97(3):375--423, 1973.
	
	\bibitem{Gro10}
	M.~Gromov.
	\newblock {Singularities, Expanders and Topology of Maps. Part 2: from
		Combinatorics to Topology Via Algebraic Isoperimetry}.
	\newblock {\em Geometric And Functional Analysis}, 20(2):416--526, 2010.
	
	\bibitem{GW12}
	A.~Gundert and U.~Wagner.
	\newblock {On Laplacians of Random Complexes}.
	\newblock In {\em Proceedings of the Twenty-eighth Annual Symposium on
		Computational Geometry}, pages 151--160. ACM, 2012.
	
	\bibitem{HLW06}
	S.~Hoory, N.~Linial, and A.~Wigderson.
	\newblock {Expander Graphs and their Applications}.
	\newblock {\em Bulletin of the American Mathematical Society}, 43(4):439--561,
	2006.
	
	\bibitem{KKL14}
	T.~Kaufman, D.~Kazhdan, and A.~Lubotzky.
	\newblock {Ramanujan Complexes and Bounded Degree Topological Expanders}.
	\newblock In {\em Foundations of Computer Science (FOCS), 2014 IEEE 55th Annual
		Symposium on}, pages 484--493, 2014.
	
	\bibitem{KL14}
	T.~Kaufman and A.~Lubotzky.
	\newblock High dimensional expanders and property testing.
	\newblock In {\em Innovations in Theoretical Computer Science, ITCS'14,
		Princeton, NJ, USA, January 12-14, 2014}, pages 501--506, 2014.
	
	\bibitem{Lov93}
	L.~Lov{\'a}sz.
	\newblock Random walks on graphs.
	\newblock {\em Combinatorics, Paul erdos is eighty}, 2:1--46, 1993.
	
	\bibitem{Lub12}
	A.~Lubotzky.
	\newblock Expander graphs in pure and applied mathematics.
	\newblock {\em Bulletin of the American Mathematical Society}, 49(1):113--162,
	2012.
	
	\bibitem{Lub14}
	A.~Lubotzky.
	\newblock Ramanujan complexes and high dimensional expanders.
	\newblock {\em Japanese Journal of Mathematics}, 9(2):137--169, 2014.
	
	\bibitem{LSV05.2}
	A.~Lubotzky, B.~Samuels, and U.~Vishne.
	\newblock Explicit constructions of ramanujan complexes of type
	$\widetilde{A}_d$.
	\newblock {\em European Journal of Combinatorics}, 26(6):965--993, 2005.
	
	\bibitem{LSV05.1}
	A.~Lubotzky, B.~Samuels, and U.~Vishne.
	\newblock Ramanujan complexes of type $\widetilde{A}_{d}$.
	\newblock {\em Israel Journal of Mathematics}, 149(1):267--299, 2005.
	
	\bibitem{Opp15}
	I.~Oppenheim.
	\newblock {Isoperimetric Inequalities and topological overlapping for quotients
		of Affine buildings}.
	\newblock arXiv:1501.04940, 2015.
	
	\bibitem{PR12}
	O.~Parzanchevski and R.~Rosenthal.
	\newblock Simplicial complexes: spectrum, homology and random walks.
	\newblock arXiv:1211.6775, 2012.
	
	\bibitem{SJ89}
	A.~Sinclair and M.~Jerrum.
	\newblock {Approximate counting, uniform generation and rapidly mixing Markov
		chains}.
	\newblock {\em Information and Computation}, 82(1):93--133, 1989.
	
	\bibitem{SS96}
	M.~Sipser and D.~A. Spielman.
	\newblock Expander codes.
	\newblock {\em IEEE Transactions on Information Theory}, 42(6):1710--1722,
	1996.
	
	\bibitem{Spi95}
	D.~A. Spielman.
	\newblock {\em Computationally efficient error-correcting codes and holographic
		proofs}.
	\newblock PhD thesis, Massachusetts Institute of Technology, 1995.
	
	\bibitem{Tre16}
	L.~Trevisan.
	\newblock {Cheeger-type Inequalities for $\lambda_n$}.
	\newblock
	\url{http://lucatrevisan.wordpress.com/2016/02/09/cheeger-type-inequalities-for-%CE%BBn/},
		2016.
		
\end{thebibliography}

\end{document}